\pdfoutput=1
\documentclass[11pt]{article}

\usepackage[final]{acl}

\usepackage{times}
\usepackage{latexsym}

\usepackage[T1]{fontenc}
\usepackage[utf8]{inputenc}

\usepackage{microtype}

\usepackage{inconsolata}

\usepackage{graphicx}
\usepackage{amsmath,amsfonts}
\usepackage{algorithmic}
\usepackage{graphicx}
\usepackage{textcomp}
\usepackage{xcolor}
\usepackage{comment}
\usepackage{algorithm}
\usepackage{algorithmic}
\usepackage{stfloats} % 在双栏文档中，stfloats 包可以让你控制跨栏浮动图的精确位置。
\usepackage{graphicx} % 支持表格缩放
\usepackage{multirow} % 支持多行单元格
\usepackage{float}    % 提供 [H] 浮动位置修正
\usepackage{array}    % 自定义列宽与对齐方式
\usepackage{subcaption} % 加载 subcaption 包
\usepackage{booktabs} % 用于更好的表格排版
\renewcommand{\arraystretch}{1.8} % 调整行高（选项）

\def\BibTeX{{\rm B\kern-.05em{\sc i\kern-.025em b}\kern-.08em
    T\kern-.1667em\lower.7ex\hbox{E}\kern-.125emX}}

\usepackage{url} %for url
\usepackage{float}
\usepackage{enumitem}
\usepackage{xcolor}
\usepackage{marvosym}
\usepackage{amsthm}

\usepackage{amsthm}

\newtheorem{proposition}{Proposition}

\usepackage{amsthm}

\theoremstyle{remark}
\newtheorem{remark}{Remark}

\renewcommand{\textcolor}[2]{\begingroup\color{black}#2\endgroup} %消除颜色标记

\title{PRA-RAG: Provably Robust Aggregation in Retrieval-Augmented Generation against Retrieval Corruption}
\author{
 \textbf{Xue Tan\textsuperscript{1,3}},
 \textbf{Yi Zheng\textsuperscript{1}},
 \textbf{Chang Huo\textsuperscript{1}},
 \textbf{Yunruo Zhang\textsuperscript{3}},
 \textbf{Yu Liu\textsuperscript{1,3}},
 \textbf{Hao Luan\textsuperscript{1,3}}, \vspace{-0.8em} \\
 \textbf{Zhuyang Yu\textsuperscript{1,3}},
 \textbf{Xiaoyan Sun\textsuperscript{2}$^{\textrm{\Letter}}$},
 \textbf{Ping Chen\textsuperscript{3}$^{\textrm{\Letter}}$},
 \textbf{Jun Dai\textsuperscript{2}}$^{\textrm{\Letter}}$ \vspace{-0.7em}
\\
 \textsuperscript{1}School of Computer Science, Fudan University, Shanghai, China, \vspace{-0.7em}
\\
 \textsuperscript{2}Department of Computer Science, Worcester Polytechnic Institute, MA, USA, \vspace{-0.7em}
\\
  \textsuperscript{3}Institute of Big Data, Fudan University, Shanghai, China, \vspace{-0.7em}
\\
\setlength{\baselineskip}{0.6\baselineskip}  % 缩短行距
 \small{
   \textbf{Corresponding authors:}    \href{mailto:pchen@fudan.edu.cn}{pchen@fudan.edu.cn,}
   \href{mailto:xsun7@wpi.edu}{xsun7@wpi.edu,}
   \href{mailto:jdai@wpi.edu}{jdai@wpi.edu}
 }
 \setlength{\baselineskip}{0.6\baselineskip}  % 缩短行距
}

\begin{document}
\maketitle
\begin{abstract}
Retrieval-Augmented Generation (RAG) enhances Large Language Models (LLMs) by incorporating external knowledge, effectively mitigating their inherent knowledge limitations. However, RAG remains vulnerable to poisoning attacks that manipulate retrieved texts to mislead model outputs. Existing defense mechanisms often lack theoretical robustness guarantees and perform unreliably when the LLM has limited knowledge of the retrieved content. In this work, we propose PRA-RAG, a provably robust retrieval aggregation algorithm designed to defend against poisoning attacks on retrieved texts. PRA-RAG samples multiple combinations of retrieved texts and utilizes geometric structures in the embedding space to identify a robust subset, from which a stable aggregated representation is derived. We provide theoretical bounds on the maximum impact of poisoned retrieved content and establish a quantitative measure of RAG’s robustness. Experiments across multiple benchmarks and RAG architectures demonstrate that PRA-RAG reduces the attack success rate to as low as \textit{1\%} while maintaining an accuracy of \textit{71\%}, significantly outperforming representative state-of-the-art (SOTA) methods. 
\end{abstract}

\maketitle

\section{Introduction}
Retrieval-Augmented Generation (RAG)~\cite{lewis2020retrieval} is an advanced generative paradigm that integrates external knowledge databases to effectively address the limitations of LLMs in domain-specific knowledge coverage and access to up-to-date information.
% which pose significant challenges to their applications in fields such as healthcare~\cite{wang2024potential}, economics~\cite{loukas2023making}, and scientific research~\cite{kumar2023mycrunchgpt, boyko2023interdisciplinary}. 
In a typical RAG pipeline, when a user submits a query (e.g., ``\textit{What is the name of the highest mountain}?"), the retriever first encodes it into an embedding vector using a text encoder (e.g., BERT~\cite{kenton2019bert}) and then retrieves the most similar texts from an external knowledge database. These retrieved texts are subsequently provided as context to the LLM to guide and enhance the response generation process. RAG has been widely adopted in various real-world applications due to its strengths in knowledge augmentation and high-quality generation, with prominent examples including ChatGPT~\cite{achiam2023gpt}, Microsoft Bing Chat~\cite{bingchat}, and Google Search AI~\cite{googlegenai}. However, the integration of external knowledge databases further exacerbates concerns regarding the security of LLMs~\cite{rocky, bbc}. 

Recent studies have shown that injecting malicious content into retrieved texts can steer the LLM to generate responses aligned with the attacker’s intent (e.g., the target answer could be ``\textit{Fuji}" when the target question is ``\textit{What is the name of the highest mountain}?"), thereby posing a serious threat to the reliability and security of RAG systems~\cite{zou2024poisonedrag, greshake2023not, tan2024glue}. 
% For example, PoisonedRAG~\cite{zou2024poisonedrag} constructs poisoned texts through optimization, ensuring they are retrieved and induce the LLM to generate the attacker’s target response. 
% AdvDec~\cite{zhang2025adversarialdecoding} crafts semantically benign poisoned texts via adversarial decoding to effectively compromise RAG systems. 
To counter attacks induced by poisoned texts, AstuteRAG~\cite{wang2025} and TrustRAG~\cite{zhou2025trustrag} 
introduce detection-based defenses that leverage the internal knowledge of LLMs to identify and filter maliciously retrieved content. However, their effectiveness is limited in scenarios where the LLM lacks sufficient knowledge to recognize the poisoned inputs. 
% INSTRUCTRAG~\cite{wei2024instructrag} trains the LLM to denoise retrieved content by learning to suppress poisoned or irrelevant information during generation. 
Moreover, existing methods~\cite{wang2025, zhou2025trustrag, wei2024instructrag, asai2024self} lack a theoretical framework for certifying or quantifying the robustness of RAG systems, leaving their reliability under adversarial conditions largely unverified. 
RobustRAG~\cite{xiang2024certifiably} guarantees robustness for LLM outputs but incurs high computational overhead by requiring multiple LLM generations.
% RobustRAG~\cite{xiang2024certifiably} offers provable robustness guarantees for the final LLM outputs. However, it requires generating responses from multiple LLMs, resulting in significant computational overhead.
% and limiting its practicality in real-world deployments.

In this work, we propose PRA-RAG, a \textbf{P}rovably \textbf{R}obust retrieval \textbf{A}ggregation algorithm for \textbf{RAG} systems, designed to mitigate the risk of poisoning attacks at the retrieval stage. 
% and enhance the stability and security of RAG systems under adversarial conditions. 
During retrieval, we intentionally expand the candidate set to increase informational diversity, which enhances the system’s resilience to poisoned content by reducing the influence of any single malicious text. The aggregation process consists of three key steps. \textbf{First}, multiple subsets of the retrieved texts are sampled to form diverse candidate combinations, thereby constraining the influence of poisoned content within a controllable range. \textbf{Second}, Each combination is encoded into an embedding vector. In this geometric space, we identify a \textit{minimum-radius ball} that encloses more than half of the combinations and take its center as the selected robust subset.
% In practice, the center of this ball corresponds to the desired robust subset and should thus be selected.
%In practice, the center of this ball is the desired robust subset that best serves as the robust aggregated representation, and should therefore be selected.
\textbf{Finally}, a weighted average of the selected subset’s embeddings yields a robust representation that captures consensus and reduces outlier influence.
% a weighted average of the embeddings from the selected subset is computed to derive the final robust representation, effectively capturing the consensus while mitigating the influence of outliers. 

% Benefiting from the proposed concept of minimum radius ball, we can model its semantic shift in retrieved texts caused by poisoning attacks and the shift is found to be subject to a theoretical upper bound, as mathematically proved in Section ``\textit{Computing the Certified Maximum Deviation}".

Benefiting from the proposed concept of the minimum radius ball, we model the semantic shift in retrieved texts caused by poisoning attacks. We show that this shift is bounded by a theoretical upper limit, as mathematically proved in Section~\ref{sec:the}. The semantic shift quantifies how poisoning influences retrieval, while the upper bound offers a theoretical basis for constraining this effect. The key contributions of our work are:
% A lower poisoning influence results in safer downstream RAG generation.
% Reducing poisoning influence improves the reliability and safety of downstream RAG generation.
% The semantic shift modulates the poisoning influence, and the upper bound provides an empirical foundation for reducing this influence.
% This upper bound on semantic shift effectively quantifies the impact of poisoning attacks on retrieved texts; a smaller bound indicates weaker attack influence and leads to safer downstream generation.By certifying the stability of the aggregated representation under bounded corruption, we provide theoretical support for the overall robustness of RAG. 
% We theoretically prove and quantify the maximum impact that poisoned texts can cause, thereby ensuring the reliability of the aggregated retrieval results when this impact remains below a predefined threshold.

% In summary, our key contributions are as follows:
% We summarize our main contributions in this paper as follows:

\begin{itemize}
    \item \textbf{Theory:} We theoretically characterize and certify the maximum semantic deviation caused by poisoned retrievals, providing a principled metric to evaluate system robustness and attack effectiveness.
    % We provide a theoretical characterization and certification of the maximum semantic deviation introduced by poisoned retrievals. This bound serves as a principled metric for assessing both system robustness and attack strength.

    \item \textbf{Algorithm:} We propose PRA-RAG, a provably robust aggregation algorithm for RAG systems that effectively mitigates the influence of poisoned retrievals and preserves the accuracy of generated outputs.

    \item \textbf{Evaluation:} We conduct extensive experiments to validate the effectiveness of PRA-RAG. For example, on the MSMARCO dataset with 20\% of the retrievals poisoned, 
    our method achieves an accuracy of 71\% while reducing the attack success rate to just 1\%. Our approach also surpasses state-of-the-art methods in efficiency.
    %Our approach also outperforms SOTA methods regarding efficiency, reducing the average response time per query by 75\% compared to RobustRAG.

    %Additionally, our approach reduces the average response time per query by 75\% compared to RobustRAG.
    % We conduct extensive experiments to validate the effectiveness of PRA-RAG. For instance, on the MSMARCO dataset with 20\% of retrievals poisoned, our method achieves an accuracy of 71\% while reducing the attack success rate to just 1\%. Furthermore, our approach reduces the average response time per query by 75\% compared to RobustRAG.
\end{itemize}

\section{Background and Related Work}

\textbf{{Retrieval-Augmented Generation.}}
% Retrieval-Augmented Generation (RAG)~\cite{guu2020retrieval, lewis2020retrieval} addresses the issue of inaccurate or outdated responses in LLMs by incorporating accurate, relevant, and up-to-date external information. A key advantage of RAG lies in its refreshable knowledge database, enabling the system to stay current without costly fine-tuning. Owing to this flexibility, RAG has seen widespread adoption in industry (e.g., WikiChat~\cite{semnani2023wikichat}, FinGPT~\cite{zhang2023enhancing}, ChatRTX~\cite{chatrtx2024}) and in specialized domains such as medical diagnostics~\cite{siriwardhana2023improving} and code/email completion~\cite{parvez2021retrieval}.
RAG comprises three components: \textit{knowledge database}, \textit{retriever}, and \textit{LLM}. The knowledge database contains newly added or updated information beyond the LLM's training data, often sourced from Wikipedia~\cite{thakur2021beir} and similar platforms. The retriever selects the Top-$K$ texts most similar to the query $q$ as external context, which the LLM uses alongside $q$ to generate an answer. RAG involves two key stages: retrieval and generation. In the retrieval step, a retriever selects the Top-$K$ relevant knowledge pieces for a query $q$. This is done using two encoders: $E_q$ for the query and $E_p$ for knowledge passages. Each passage embedding $E_p(p_i)$ is compared with $E_q(q)$ using a similarity metric (e.g., cosine or dot product), and the Top-$K$ results form the context $\mathcal{X}_q$. In the generation step, the query $q$ and context $\mathcal{X}_q$ are combined as a prompt for the LLM to generate a response.
% In the generation step, the query $q$ and context $\mathcal{X}_q$ are combined into a prompt, and the LLM generates a response $\text{LLM}(q, \mathcal{X}_q)$.
\\\textbf{Retrieval Corruption Attack.} The use of external knowledge databases in RAG introduces security vulnerabilities.
% as adversaries can exploit them to inject misinformation and indirectly influence the model’s outputs~\cite{zou2024poisonedrag, xue2024badrag, jiao2024exploring}.
PoisonedRAG~\cite{zou2024poisonedrag} generates adversarially crafted texts through an optimization-based approach and injects them into the knowledge database, thereby inducing the LLM to produce attacker-specified target responses. The Denial-of-Service Attack~\cite{shafran2024machine} involves inserting a single ``blocker" document into the knowledge database, which is retrieved in response to a specific query and causes the RAG system to refuse to answer that query. Adversarial Decoding~\cite{zhang2025adversarialdecoding} performs RAG poisoning and LLM guard evasion by generating readable adversarial texts through adversarial decoding. 
% PRCAP~\cite{zhong2023poisoning} perturbs discrete tokens to craft adversarial samples similar to training queries and injects them into the knowledge database. 
BadRAG~\cite{xue2024badrag} employs white-box optimization to attack the retriever and uses handcrafted documents to target the generator. \textcolor{blue}{PR-Attack~\cite{jiao2025pr} jointly optimizes a trigger and a small set of poisoned texts to stealthily induce the RAG system to generate a target response. CorruptRAG~\cite{zhang2025practical} misleads the model into producing targeted false content by prompting the LLM to label the correct answer as outdated and to generate a fabricated latest but incorrect answer.}
The aforementioned attack methods are highly effective. To assess our defense, we select representative attacks to simulate realistic RAG poisoning scenarios.
% The aforementioned attack methods have demonstrated significant effectiveness. To evaluate the proposed defense strategy, we select representative attacks to simulate real-world RAG poisoning scenarios.
\\\textbf{The Robustness of RAG.} In order to defend against the aforementioned retrieval-targeted poisoning attacks, 
TrustRAG~\cite{zhou2025trustrag} introduces a two-stage defense mechanism that leverages the semantic characteristics of poisoned texts and the inherent knowledge of the LLM, effectively mitigating both single-point and multi-corpus injection attacks. 
INSTRUCTRAG~\cite{wei2024instructrag} is designed to explicitly learn how to denoise retrieved content, thereby addressing the presence of poisoned or irrelevant information. SELF-RAG~\cite{asai2024self} proposes a framework that enhances generation quality and factual accuracy through self-reflection mechanisms within the LLM. \textcolor{blue}{RAGForensics~\cite{zhang2025traceback} employs an iterative retrieval mechanism that combines an LLM with carefully designed chain-of-thought prompts to perform round-by-round judgment and filtering of candidate retrieved texts. RAGuard~\cite{cheng2025secure} integrates a two-stage filtering mechanism based on chunk-level perplexity and text similarity, which can effectively identify and remove malicious poisoned content from the retrieved results.}
Other strategies include carefully crafted prompts~\cite{cho2023improving, press2023measuring}, plug-in model architectures~\cite{baek2023knowledge}, and purpose-built specialized models~\cite{yoran2023making}. However, these defense strategies lack quantitative robustness evaluations and theoretical guarantees for RAG systems.
% However, the aforementioned defense strategies have neither provided quantitative evaluations of RAG’s robustness nor established theoretical guarantees. 
% RobustRAG~\cite{xiang2024certifiably} improves RAG robustness by using multiple LLMs to generate diverse responses and applying a voting mechanism to mitigate poisoned content. While it offers theoretical guarantees, the approach incurs high computational cost and poses challenges in aggregating outputs. 
RobustRAG~\cite{xiang2024certifiably} enhances RAG robustness by using multiple LLMs and a voting mechanism to filter poisoned content, providing theoretical guarantees but at the cost of significant computational overhead and challenging output aggregation.
We propose a provably robust RAG algorithm that quantitatively evaluates robustness and provides strong empirical defense against poisoning attacks.
% To address these issues, we propose a provably robust RAG algorithm that enables quantitative robustness evaluation and demonstrates strong empirical defense against poisoning attacks.

\begin{figure*}[t]
    \centering
    \includegraphics[width=1.0\linewidth]{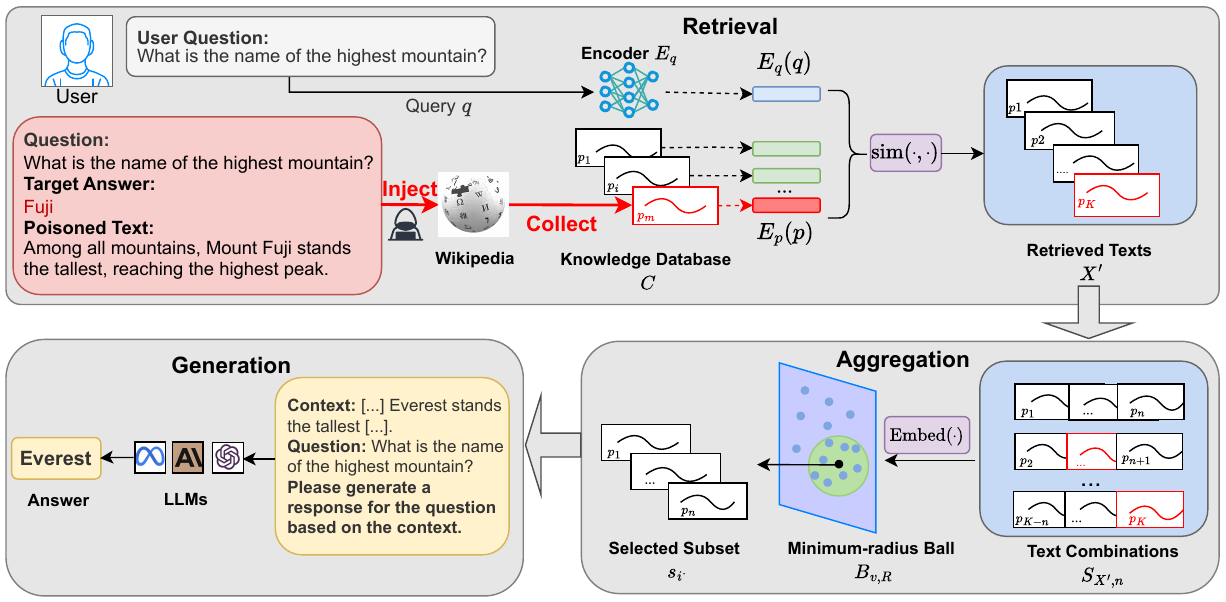}
    \caption{The PRA-RAG pipeline against poisoned retrieval.}
    \label{fig:workflow}
\end{figure*}

% 把这个球标识出来

\section{Preliminary}

\subsection{Threat Model}
We study \emph{retrieval-corruption} attacks on RAG systems, where an adversary poisons the external corpus so that the retrieved context steers the generator to produce attacker-chosen outputs. We first specify the system and notation, then state the attacker’s objectives and capabilities. 
% and finally list the assumptions used in our analysis.
% We investigate \emph{retrieval-corruption} attacks on RAG systems, where adversaries poison the external corpus to manipulate retrieved contexts and induce attacker-specified outputs. We begin by defining the system and notation, followed by outlining the attacker’s objectives and capabilities.

\paragraph{System and notation.}
Let $\mathcal{D}$ denote the corpus, and let $\mathcal{R}$ be the retriever that returns the Top-$K$ passages for a query $q$, denoted as $\mathcal{X}_K(q;\mathcal{D})=\{p_1,\ldots,p_K\}$. Under poisoning, the adversary injects a set of passages $\Gamma$ into the corpus, yielding $\mathcal{D}'=\mathcal{D}\cup\Gamma$. The corresponding retrieved set is $\mathcal{X}'_K(q)=\mathcal{X}_K(q;\mathcal{D}')$ with $\varepsilon \;=\; \mathcal{X}'_K(q)\cap\Gamma$ and $\mathcal{X}'_K(q)=\{\tilde p_1,\ldots,\tilde p_{\varepsilon}\}\cup\{p_1,\ldots,p_{K-\varepsilon}\}$. We use $n$ to denote the subset size in our mechanism.
% We denote by $n$ the subset size used later by our aggregation mechanism.

\paragraph{Attacker’s objectives.}  
The attacker selects a set of target queries $Q = \{q_1, \ldots, q_M\}$ along with corresponding target answers $A = \{a_1, \ldots, a_M\}$. 
% The objective is to poison the corpus such that for each $q_i \in Q$, the RAG system, when conditioned on the retrieved context from the contaminated corpus $\mathcal{D}'$, generates the attacker-specified answer $a_i$. 
The goal is to poison the corpus so that, for each $q_i \in Q$, the RAG system generates the attacker-specified answer $a_i$ when using retrieved context from the contaminated corpus $\mathcal{D}'$.
For instance, given the query $q_i$: \emph{``What is the name of the highest mountain?''}, the attacker attempts to make the system output the incorrect response \emph{``Mount Fuji''}.

\paragraph{Attacker’s capabilities.}
We consider an adversary who can inject a set of malicious passages $\Gamma=\{\tilde p_j^i \mid i=1,\ldots,M;\; j=1,\ldots,N\}$ into the external corpus (e.g., via content published on open platforms such as Wikipedia or via data vendors that aggregate third-party content). The attacker has no knowledge of the LLM parameters or decoding, but is assumed to have white-box knowledge of the retriever because many retrievers (e.g., Contriever~\cite{izacard2021unsupervised}, ANCE~\cite{xiong2020approximate}) are publicly available. We do not assume strict geometric separability between clean and poisoned passages. Our method’s robustness is distribution-free and relies solely on a \emph{combinatorial majority}: among the Top-$K$ items with at most $\varepsilon$ poisons, the $\binom{K-\varepsilon}{n}$ clean size-$n$ subsets exceed half of all $\binom{K}{n}$ subsets whenever $\binom{K-\varepsilon}{n}>\tfrac12\binom{K}{n}$. Our aggregation targets this majority, so certification holds even when embeddings are very close. This assumption is valid in real-world scenarios, as achieving an excessively large $\varepsilon$ to heavily poison a RAG database incurs prohibitive costs and is inherently challenging for attackers.

% We analyze the practical regime where the retrieved set contains $\varepsilon$ poisoned items with $2\varepsilon < K$; achieving $2\varepsilon \ge K$ incurs prohibitive cost and makes answering difficult even for humans~\cite{xiang2024certifiably}.

%(e.g., if $2n<K-\varepsilon$)
\subsection{Defense Objective}
In this work, our defense primarily focuses on mitigating attacks launched through the injection of poisoned texts~\cite{zou2024poisonedrag, xue2024badrag, jiao2024exploring} rather than prompt injection~\cite{greshake2023not}. Our approach aims to limit the impact of poisoned passages while maintaining the retrieval and generation quality in RAG.
% Our approach is designed to bound and reduce the impact of poisoned passages while preserving both retrieval and generation quality in RAG.

\paragraph{Quantifying the effect of poisoning.}
Heuristic defenses often lack a precise characterization of how poisoned content perturbs the system. Our robustness operator $\mathcal{F}$ aggregates the retrieved set $\mathcal{X}'$ and provides a certified upper bound on the induced semantic deviation (PAD), enabling quantitative evaluation of RAG robustness.

\paragraph{Low ASR with high ACC.}
Mitigating poisoning while maintaining usability is essential. A low Attack Success Rate (ASR) indicates effective defense, whereas a high Accuracy (ACC) reflects preserved generation quality. We formalize the robustness–utility trade-off as
\begin{equation}
    \begin{aligned}
\min_{\mathcal{F}} & \quad\mathrm{ASR}(\mathrm{LLM}(q_i;{\mathcal{F}}(\mathcal{X}'_K(q_i))) \\
\mathrm{s.t.} & \quad\mathrm{ACC}(\mathrm{LLM}(q_i;{\mathcal{F}}(\mathcal{X}'_K(q_i)))\geq\theta.
\end{aligned}
\end{equation}\\
where $\theta$ is a user-specified minimum acceptable accuracy. Overly conservative behavior (e.g., blanket refusals) may trivially reduce ASR but harms usability; our objective explicitly discourages such solutions by enforcing the ACC constraint.
% This algorithm effectively filters out poisoned content while preserving the diversity and utility of retrieved texts. Additionally, it allows for the quantitative evaluation of the impact of poisoned content.
\section{Our Provable PRA-RAG}

% In RAG, external corpora may contain unreliable or unverified information; naively conditioning on retrieved passages can lead to inaccurate or harmful outputs~\cite{zou2024poisonedrag,xue2024badrag,jiao2024exploring}.
In RAG, we cast retrieval corruption as a \emph{set perturbation} problem, where an adversary alters at most $\varepsilon$ items in the Top-$K$ set, and propose a geometric, model-agnostic retrieval aggregation that operates on \emph{combination embeddings}. 
% By selecting a majority-covering center, the method yields \emph{certifiable} bounds on semantic shift and relies on a \emph{combinatorial majority} of clean subsets rather than margin-based separability. We formalize this view and demonstrate both theoretical guarantees and practical benefits.
As illustrated in Fig.~\ref{fig:workflow}, given a query $q$, we retrieve the Top-$K$ most relevant documents from a knowledge database, denoted as the set $\mathcal{X}=\{p_1,\ldots,p_K\}$. In our approach, we set the number of retrieved texts $K$ slightly higher than the actual requirement to introduce diversity, thereby mitigating the impact of poisoned content. Next, we sample multiple subsets from the retrieved texts. Let $\mathcal{S}_{\mathcal{X}, n} = \{s_1, s_2, \ldots, s_{L}\}$ denote the set of all possible subsets of size $n$ drawn from the $K$ retrieved documents, where $L = \binom{K}{n}$. Each document $p_i$ in the subset is encoded as an embedding vector $\mathbf{e}_i \gets Embed(p_i), \mathbf{e}_i \in \mathbb{R}^d$, yielding a set of embeddings in the $d$-dimensional vector space. We construct the corresponding vector set $\mathcal{V}_{\mathcal{X}, n} = \{v_1, v_2, \ldots, v_{L}\}$ from the set $\mathcal{S}$, where each vector $v_i$ is formed by concatenating the embedding vectors $\mathbf{e}_i$ of all texts in $s_i$. The concatenation order of $\mathbf{e}_i$ within $s_i$ strictly follows the relative order of their corresponding texts $p_i$.

\begin{algorithm}[t]
\caption{Inference Procedure for Selecting Robust Retrieval Text}
\label{alg:inference}
\begin{algorithmic}[1]
\STATE \textbf{Input}: A set of Top-$K$ retrieved texts $\mathcal{X} = \{p_1, p_2, \dots, p_K\}$; subset size $n$ (with constraint $2n < K$).\\
\STATE \textbf{Output}: A robustly aggregated retrieval text $x^*$ for LLMs.
% \STATE Perform $\binom{K}{n}$ combinatorial samplings on the $K$ retrieved texts to obtain a collection of text subsets $\mathcal{S}_{\mathcal{X}, n} = \{s_1, s_2, \ldots, s_L\}$, where $L = \binom{K}{n}$.
% \STATE Given retrieved texts $X = \{p_1, p_2, \ldots, p_K\}$
% \STATE $\mathcal{S}_{\mathcal{X}, n} \leftarrow$ all $n$-combinations of $X$
\STATE $\mathcal{S}_{\mathcal{X}, n} = \{ s_1, s_2, \ldots, s_L \} \leftarrow \mathrm{Comb}(\mathcal{X}, n)$, $L = \binom{K}{n}$
\FOR{$i = 1$ to $L$}
    \STATE $\text{embeddings} \leftarrow [\, Embed(p) \mid p \in s_i\,]$
    \STATE $v_i \leftarrow \mathrm{Concat}(\text{embeddings})$
    \STATE $\mathcal{V}_{\mathcal{X}, n} \leftarrow \mathcal{V}_{\mathcal{X}, n} \cup \{v_i\}$
\ENDFOR

\FOR{$i = 1$ to $L$}
    \STATE \textcolor{blue}{Initialize similarity vector $\mathbf{sim}_i \in {R}^{L-1}$}
    \FOR{$j = 1$ to $L$ and $j \ne i$}
        \STATE \textcolor{blue}{$d_{ij} \gets \arccos(\cos(v_i, v_j))$}
        \STATE Store $d_{ij}$ in $\mathbf{sim}_i$
    \ENDFOR
    \STATE Sort $\mathbf{sim}_i$ in ascending order
    \STATE \textcolor{blue}{Let $m_i \gets$ the $k$-th smallest value in $\mathbf{sim}_i$, $k = \left\lfloor \frac{L}{2} \right\rfloor$} \hfill  
    % // Robust similarity score
\ENDFOR
\STATE Identify $i^* = \arg\min_{i} m_i$, $R=m_{i^*}$
\STATE Compute $x_{i^*} = \mathrm{Aggregate}([\, Embed(p) \mid p \in s_{i^*} \,])$
% \STATE Compute 
% \[
% x_{i^*} = \frac{\sum\limits_{p \in s_{i^*}} \mathrm{sim}(p, q) \cdot \mathrm{Embed}(p)}{\sum\limits_{p \in s_{i^*}} \mathrm{sim}(p, q)}
% \]
% \STATE Identify $i^\ast = \arg\min_{i} m_i$
% \STATE Let $s_{i^\ast}$ be the corresponding $n$ data items
% % \STATE Average the embedding vectors of subset $s_{i^*}$ to obtain the final robust representation $x_{i^*}$.
% \STATE $\text{embeddings} \leftarrow [\, Embed(p) \mid p \in s_{i^*} \,]$
% \STATE $x_{i^*} \leftarrow \mathrm{Mean}(\text{embeddings})$
\STATE \textbf{return} $x_{i^\ast}$ as input to the LLMs
\end{algorithmic}
\end{algorithm}

\textcolor{blue}{We define the distance between two vectors $u, v \in \mathcal{V}_{\mathcal{X},n}$ as the angular distance: $d(u, v) = \arccos(\cos(u, v)) \in [0, \pi]$, where $\cos(u, v) = \frac{u \cdot v}{\|u\|\|v\|}$ denotes the cosine similarity.} 
% This is symmetric and satisfies the triangle inequality, although it may not be a strict metric. 
For a point $v \in \mathcal{V}_{\mathcal{X},n}$ and radius $R > 0$, define the ball: 
% $B(v, R) = \{ u \in \mathcal{V}_{\mathcal{X},n} : d(u, v) \leq R \}$.
\begin{equation}
 B(v, R) = \{ u \in \mathcal{V}_{\mathcal{X},n} : d(u, v) \leq R \}.
\end{equation}
Our method $\mathcal{F}$ aims to identify the center of a minimum-radius ball that contains more than half of the points, representing the desired output of the robust aggregation process.
% Our method $\mathcal{F}$ aims to identify the center of a minimum-radius ball, which represents the desired output of the robust aggregation process. 
% Since finding the minimum-radius enclosing ball in the entire space is an NP-hard problem, we adopt an approximate strategy by considering each point in the set $V_{\mathcal{X},n}$ as a candidate center of the ball~\cite{kumar2021center}.
% This center is selected to maximize resistance against adversarial perturbations. 
Formally, the objective is defined as:
\begin{equation}
\begin{aligned}
\mathcal{F}(\mathcal{X}) = \arg\min_{z} \; R \quad \text{s.t.} \quad 
& \sum_{w \in V_{\mathcal{X}, n}} \mathbf{1}_{\mathcal{B}(z, R)}(w) \\
& \geq \left\lfloor \frac{\binom{K}{n}}{2} \right\rfloor + 1,
\end{aligned}
\end{equation}
where $\mathcal{X}$ denotes the input set, $V_{\mathcal{X},n}$ is the vector collection of all $n$-element subsets of $\mathcal{X}$, $\mathcal{B}(z, R)$ denotes a ball of radius $R$ centered at $z$. The indicator function $\mathbf{1}_{\mathcal{B}(z, R)}(w)$ is defined to be $1$ if $w \in \mathcal{B}(z, R)$, and $0$ otherwise.
% \begin{equation}
%     \mathbf{1}_{\mathcal{B}(z, R)}(w)=\left\{\begin{array}{l}
% 1, \text { if } w \in \mathcal{B}(z,R) \\
% 0, \text { otherwise }.
% \end{array}\right.
% \end{equation}
This formulation ensures that the selected center $z$ lies within a ball that encompasses the majority of the subsets.
Since $S_{\mathcal{X},n}$ includes all possible combinations of the retrieved texts, the centroid $z$ corresponding to the selected subset $s$ effectively captures the meaning of the majority of texts in $\mathcal{X}$, thereby filtering out a small number of poisoned texts. 

Finally, we compute a weighted average of the embedding vectors of the texts in the selected subset $s_{i^*}$, where the weights are determined by the similarity (e.g., cosine similarity) between each text $p \in s_{i^*}$ and the query $q$:
\begin{equation}
x_{i^*} = \frac{\sum\limits_{p \in s_{i^*}} \mathrm{sim}(p, q) \cdot \mathrm{Embed}(p)}{\sum\limits_{p \in s_{i^*}} \mathrm{sim}(p, q)},
\end{equation}
which helps prevent poisoned texts from evading filtration due to semantic similarity with clean texts, yielding the final robust aggregation result. More details are provided in Algorithm~\ref{alg:inference}. 

% In white-box LLM scenarios, we can directly access the model’s embedding module to extract vector representations of the retrieved texts and perform robust aggregation before generating the final response. For black-box LLMs, where the embedding layer is inaccessible, we use open-source models to obtain the embeddings of the retrieved texts, aggregate them, and then decode the aggregated vector back into natural language text to serve as context input for the black-box model. This design ensures the applicability and flexibility of our method across various LLM deployment settings.

\section{Computing the Certified Maximum Deviation}
\label{sec:the}
In the context of retrieval-based text poisoning attacks on RAG systems, we assume that the attacker successfully injects $\varepsilon$ malicious documents into the final retrieved set $\mathcal{X}$, resulting in a perturbed set $\mathcal{X^{'} } = \left [ \tilde{p}_{1}, \tilde{p}_{2}, \dots, \tilde{p}_{\varepsilon } \right ] \cup \left [ {p_1, p_2, \dots, p_{K-\varepsilon} } \right ]$, where each $\tilde{p}_{i}$ represents a poisoned document inserted by the attacker, and each $p_i$ denotes a clean document. In the absence of poisoning attacks, the center of the minimum-radius ball computed by Algorithm~\ref{alg:inference} is used as the robust output. Under poisoning attacks, the radius $ \hat{R}$ is enlarged to counter the shift caused by the $\varepsilon$ poisoned texts, ensuring that over half of the embeddings remain within the ball. Accordingly, the radius is approximately set to $ \hat{R} \gets \mathcal{D}[k] $, where  
$k = \left\lfloor \frac{\binom{K}{n}}{2} \right\rfloor + \left( \binom{K}{n} - \binom{K-\varepsilon}{n} \right),$ as shown in Algorithm~\ref{alg:certify}. 
\textcolor{blue}{To ensure that the empirical radius $\hat{R}$ satisfies the theoretical majority condition, Lemma 1 (Section~\ref{lemma1}) guarantees that the majority ball contains a strict majority of clean combinations. Combined with the geometric separability assumption (Section~\ref{app:separability}), this ensures the center is anchored to the clean subset.} We denote the embedding shift of the final aggregated text caused by poisoned samples as $ d $, where $d = (1 + \beta) \hat{R}$. The proposed metric $d$ effectively quantifies the impact of poisoned texts on the retrieval results. A smaller $d$ indicates lower sensitivity to perturbations, implying stronger robustness and greater resistance to successful attacks.\\
\textbf{Certifying the Upper Bound of Deviation $d$.} Let $\mathcal{X'}$ be any corrupted version of $\mathcal{X}$ obtained by changing at
most $\varepsilon$ passages, that is $\|\mathcal{X'}-\mathcal{X}\|_{0}\le\varepsilon$.
Define the robustness radius $R$ as the smallest value satisfying
\begin{equation}
  \label{eq:R-def}
  \begin{aligned}
    \forall\,\mathcal{X'} \text{ s.t. } \|\mathcal{X'} - \mathcal{X}\|_{0} \le \varepsilon,\quad
    &\sum_{w \in \mathcal{V}_{\mathcal{X'},n}}
    \mathbf{1}_{\mathcal{B}(\mathcal{F}(\mathcal{X}), R)}(w) \\
    &\ge \left \lfloor \frac{\binom{K}{n} }{2}  \right \rfloor + 1.
  \end{aligned}
\end{equation}
% $\|\mathcal{X}' - \mathcal{X}\|_0$ denotes the number of poisoned passages. 
We present the following theorem:\\
\textbf{Theorem 1.} For $\mathcal{X'}$ satisfying $\|\mathcal{X'}-\mathcal{X}\|_{0}\le\varepsilon$,
\begin{equation}
  d\bigl(\mathcal{F}(\mathcal{X}),\,\mathcal{F}(\mathcal{X'})\bigr)\;\le\;2R, 
\end{equation}
where the proof is provided in Section~\ref{sec:proof1}.
However, computing $\mathcal{F}(\mathcal{X})$ exactly is computationally expensive in practice. Therefore, we use an approximation method from Center Smoothing~\cite{kumar2021center} to compute a $\beta$-MEB (Minimum Enclosing Ball) that contains the majority of elements in $\mathcal{V}_{\mathcal{X},n}$, denoted as:
\begin{equation}
  \label{eq:Fhat-def}
  \widehat{\mathcal{F}}(\mathcal{X}) =
  \begin{aligned}[t]
    \arg \min_{z \in \mathcal{V}_{\mathcal{X},n}} \; R \quad \text{s.t.} \quad
    &\sum_{w \in \mathcal{V}_{\mathcal{X},n}} \mathbf{1}_{\mathcal{B}(z,R)}(w) \\
    & \ge \left \lfloor \frac{\binom{K}{n}  }{2}  \right \rfloor + 1.
  \end{aligned}
\end{equation}
The radius of this approximation differs from the optimal radius by a multiplicative factor $\beta$. Similarly, the approximation of $R$, denoted as $\widehat{R}$, is defined to satisfy:

\begin{equation}
  \label{eq:robust-cert}
  \begin{aligned}
    \forall\, \mathcal{X}' \text{ s.t. } \|\mathcal{X}' - \mathcal{X}\|_0 \le \varepsilon, \quad 
    &\sum_{w \in \mathcal{V}_{\mathcal{X}',n}} 
    \mathbf{1}_{\mathcal{B}(\mathcal{F}(\mathcal{X}), \widehat{R})}(w) \\
    &\ge \left \lfloor \frac{\binom{K}{n}  }{2}  \right \rfloor + 1.
  \end{aligned}
\end{equation}
Then, we have the following theorem:\\
\textbf{Theorem 2.} For any $\mathcal{X}'$ satisfying $\|\mathcal{X}' - \mathcal{X}\|_0 \le \varepsilon$, we have:
\begin{equation}
    d(\mathcal{F}(\mathcal{X}), \mathcal{F}(\mathcal{X}')) \le (1 + \beta)\widehat{R},
\end{equation}
where the proof is provided in Section~\ref{sec:proof2}. We follow the practice in the work~\cite{kumar2021center} to use $\beta = 2$ as an approximation when computing the minimum enclosing ball:
\begin{equation}
    d(\mathcal{F}(\mathcal{X}), \mathcal{F}(\mathcal{X}')) \le 3\widehat{R}.
\end{equation}

\begin{algorithm}[t]
\caption{Certify Robustness under Poisoning}
\label{alg:certify}
\begin{algorithmic}[1]
\STATE \textbf{Input}: A set of Top-$K$ retrieved texts $\mathcal{X}$; subset size $n$ (with $2n < K$); maximum number of poisoned texts $\varepsilon$ (with $\binom{K}{n}< 2\binom{K-\varepsilon}{n}$).\\
\STATE \textbf{Output}: Certified maximum deviation $d$.
% \STATE A sample set $ \mathcal{S}_{\mathcal{X'}, n} $ of size $ L = \binom{K}{n} $ is sample from the retrieved set $ \mathcal{X'} $, and an vector set $\mathcal{V}_{\mathcal{X'}, n} = \{v_1, v_2, \ldots, v_{L}\}$ is computed for each subset in $ \mathcal{S}_{\mathcal{X'}, n} $.
\STATE Let $\mathcal{S}_{\mathcal{X'}, n} = \mathrm{Comb}(\mathcal{X'}, n)$
\STATE Initialize $\mathcal{V}_{\mathcal{X'}, n} \leftarrow \emptyset$
\FOR{each subset $s_i \in \mathcal{S}_{\mathcal{X'}, n}$}
    \STATE $v_i \leftarrow \mathrm{Concat}([\, Embed(p) \mid p \in s_i\,])$
    \STATE $\mathcal{V}_{\mathcal{X'}, n} \leftarrow \mathcal{V}_{\mathcal{X'}, n} \cup \{v_i\}$
\ENDFOR
\STATE Let $v_{\text{orig}}$ be the embedding vector of the subset selected by the Algorithm~\ref{alg:inference}.
\STATE Initialize an empty list of distances $\mathcal{D} \gets \emptyset$
\FOR{each embedding $v_i \in \mathcal{V}_{\mathcal{X'}, n}$ }
    \STATE $\textcolor{blue}{d_i \gets \arccos({\cos(v_{\text{orig}}, v_i)}))}$ \hfill 
    \STATE Append $d_i$ to $\mathcal{D}$
\ENDFOR
\STATE Sort $\mathcal{D}$ in ascending order
\STATE Let $\hat{R}  \gets \mathcal{D}[k], k = \left \lfloor \frac{\binom{K}{n}  }{2}  \right \rfloor +\left ( \binom{K}{n}  - \binom{K-\varepsilon}{n}  \right )$ 
\STATE Let $d = \left ( 1+\beta  \right ) \hat{R} $
\STATE \textbf{return} $d$ as the certified maximum deviation
\end{algorithmic}
\end{algorithm}

\textcolor{blue}{Note that while the theoretical formulation is based on the full set of $L = \binom{K}{n}$ combinations, our framework naturally generalizes to a Monte Carlo sampling approach for larger values of $K$ and $n$. In such cases, we sample $m \ll L$ subsets, which provides a statistically bounded approximation of the optimal robustness (see Section~\ref{sec:appendix_sampling_proof}) and our experimental results further corroborate its effectiveness (see Section~\ref{sec:Monte Carlo}).}

\begin{table*}[t]
    \centering
    \caption{Performance of PRA-RAG across different attack strategies and datasets under various LLMs (Top-$K=8$, subset size $n=3$, poisoning rate = 20\%). Attack abbreviations: PR = PoisonedRAG, DoS = DoS Attack, AD = AdvDec, CR = CorruptRAG, Cln = Clean.}
    \label{table: main_results}
    \renewcommand{\arraystretch}{1.15}
    \setlength{\tabcolsep}{3.5pt}
    \resizebox{1.0\linewidth}{!}{
    \begin{tabular}{ll|ccccc|ccccc|ccccc}
    \toprule
    & & \multicolumn{5}{c|}{\textbf{NQ}} 
    & \multicolumn{5}{c|}{\textbf{MS-MARCO}}  
    & \multicolumn{5}{c}{\textbf{HotpotQA}}\\
    \cmidrule(lr){3-7} \cmidrule(lr){8-12} \cmidrule(lr){13-17}
    \textbf{Models} & \textbf{Met.}
    & PR & DoS & AD & CR & Cln
    & PR & DoS & AD & CR & Cln
    & PR & DoS & AD & CR & Cln \\
    \midrule
    \multirow{3}{*}{Mistral-7B}
    & ACC & 0.62 & 0.57 & 0.55 & 0.55 & 0.60 & 0.63 & 0.65 & 0.61 & 0.59 & 0.57 & 0.37 & 0.33 & 0.39 & 0.41 & 0.38 \\ 
    & ASR & 0.01 & 0.02 & 0.02 & 0.04 & - & 0.01 & 0.03 & 0.02 & 0.00 & - & 0.03 & 0.01 & 0.02 & 0.03 & - \\ 
    & PAD & 1.23 & 1.43 & 1.30 & 1.21 & 0.79 & 1.37 & 1.46 & 1.51 & 1.39 & 0.90 & 1.33 & 1.55 & 1.43 & 1.31 & 0.83 \\
    \hline
    \multirow{3}{*}{Llama3-8B}
    & ACC & 0.47 & 0.49 & 0.52 & 0.49 & 0.50 & 0.69 & 0.71 & 0.68 & 0.67 & 0.63 & 0.30 & 0.33 & 0.31 & 0.30 & 0.29 \\
    & ASR & 0.01 & 0.02 & 0.01 & 0.03 & - & 0.03 & 0.06 & 0.05 & 0.00 & - & 0.04 & 0.05 & 0.04 & 0.04 & - \\ 
    & PAD & 0.95 & 1.04 & 0.97 & 0.93 & 0.60 & 1.09 & 1.19 & 1.06 & 1.11 & 0.75 & 1.06 & 1.17 & 1.07 & 1.04 & 0.65 \\
    \hline
    \multirow{3}{*}{Vicuna-7B}
    & ACC & 0.49 & 0.52 & 0.51 & 0.51 & 0.53 & 0.68 & 0.72 & 0.70 & 0.71 & 0.71 & 0.49 & 0.50 & 0.48 & 0.47 & 0.51 \\ 
    & ASR & 0.01 & 0.03 & 0.02 & 0.00 & - & 0.01 & 0.02 & 0.03 & 0.00 & - & 0.00 & 0.02 & 0.02 & 0.04 & - \\ 
    & PAD & 1.74 & 2.22 & 1.84 & 1.73 & 1.16 & 1.91 & 1.95 & 1.99 & 1.90 & 0.75 & 1.87 & 2.34 & 1.98 & 1.85 & 0.21 \\
    \hline
    \multirow{3}{*}{Qwen3-8B}
    & ACC & 0.53 & 0.61 & 0.64 & 0.55 & 0.64 & 0.69 & 0.65 & 0.67 & 0.63 & 0.68 & 0.49 & 0.45 & 0.40 & 0.42 & 0.45 \\ 
    & ASR & 0.06 & 0.03 & 0.02 & 0.01 & - & 0.00 & 0.00 & 0.01 & 0.00 & - & 0.09 & 0.05 & 0.02 & 0.04 & - \\ 
    & PAD & 1.01 & 0.91 & 1.05 & 0.98 & 0.68 & 0.95 & 0.89 & 0.98 & 1.01 & 0.71 & 0.95 & 0.88 & 0.91 & 0.84 & 0.62 \\
    \hline
    \multirow{3}{*}{\shortstack[l]{GPT-3.5\\-turbo}}
    & ACC & 0.48 & 0.52 & 0.52 & 0.45 & 0.47 & 0.75 & 0.75 & 0.74 & 0.70 & 0.78 & 0.40 & 0.40 & 0.38 & 0.35 & 0.39 \\ 
    & ASR & 0.00 & 0.02 & 0.00 & 0.03 & - & 0.01 & 0.00 & 0.02 & 0.00 & - & 0.07 & 0.08 & 0.02 & 0.06 & - \\ 
    & PAD & 0.95 & 1.04 & 0.97 & 0.93 & 0.61 & 1.09 & 1.19 & 1.13 & 1.11 & 0.75 & 1.06 & 1.17 & 1.07 & 1.04 & 0.65 \\ 
    \hline
    \multirow{3}{*}{\shortstack[l]{GPT-5\\-mini}}
    & ACC & 0.55 & 0.47 & 0.51 & 0.47 & 0.51 & 0.59 & 0.53 & 0.54 & 0.49 & 0.58 & 0.43 & 0.41 & 0.41 & 0.43 & 0.49 \\ 
    & ASR & 0.04 & 0.03 & 0.03 & 0.02 & - & 0.00 & 0.00 & 0.00 & 0.01 & - & 0.09 & 0.07 & 0.08 & 0.03 & - \\ 
    & PAD & 0.93 & 0.98 & 0.91 & 1.05 & 0.62 & 1.11 & 1.24 & 1.17 & 1.21 & 0.69 & 1.04 & 0.96 & 1.10 & 1.05 & 0.72 \\ 
    \bottomrule
    \end{tabular}}
\end{table*}

\section{Evaluation}
\subsection{Experimental Setup}
In this section, we provide a detailed description of the experimental setup, with additional information and default configurations presented in Section~\ref{sec:setup}.\\
\textbf{Datasets.} 
To evaluate the effectiveness of the proposed method,
we utilize three question-answering datasets: Natural Questions (NQ)~\cite{kwiatkowski2019natural}, MS-MARCO~\cite{bajaj2016ms}, and HotpotQA~\cite{yang2018hotpotqa}.\\
\textbf{LLM.} We evaluate PRA-RAG with different LLMs, including Mistral-7B~\cite{jiang2023mistral}, Llama3-8B, Vicuna-7B~\cite{chiang2023vicuna}, Qwen3-8B~\cite{yang2025qwen3}, GPT-3.5-Turbo~\cite{brown2020language}, and GPT-5-mini. \\
% Furthermore, in our method, the retriever is used to encode retrieved texts into embedding vectors. Therefore, we evaluate our method across a range of widely used retrievers, including Contriever~\cite{izacard2021unsupervised} (pre-trained), Contriever-ms (fine-tuned on MS-MARCO)~\cite{izacard2021unsupervised}, DPR-mul~\cite{karpukhin2020dense} (trained on multiple datasets), and ANCE~\cite{xiong2020approximate} (trained on MS-MARCO). \\
\textbf{Attackers.} We conduct a systematic evaluation of PRA-RAG's defense effectiveness against different types of poisoning attacks targeting the retrieved texts, following prior works~\cite{xiang2024certifiably, zhou2025trustrag}:(1) Corpus Poisoning Attack: PoisonedRAG~\cite{zou2024poisonedrag}, (2) Adversarial Decoding: AdvDec~\cite{zhang2025adversarialdecoding}, (3) Denial-Of-Service Attack: DoS Attack~\cite{shafran2024machine}, \textcolor{blue}{and (4) CorruptRAG~\cite{zhang2025practical}.} Further details of these attack methods are provided in Section~\ref{sec:poi}.
% In our experiments, we directly insert the poisoned text most similar to the query into the final set of retrieved texts to ensure that the poisoned texts effectively interfere with the retrieval results.
\\
\textbf{Defenders.} 
% Although various defense methods have been proposed to counter attacks on RAG, we highlight the advantages of our approach by comparing it with several representative state-of-the-art baselines,
We demonstrate the advantages of our approach by comparing it with representative state-of-the-art baselines, including RobustRAG~\cite{xiang2024certifiably}, InstructRAG~\cite{wei2024instructrag}, AstuteRAG~\cite{wang2025}, \textcolor{blue}{ TrustRAG~\cite{zhou2025trustrag}, and RAGForensics~\cite{zhang2025traceback}} with their default configurations. 
% to ensure defense effectiveness.
% , and TrustRAG~\cite{zhou2025trustrag}.
\\
\textbf{Evaluation Metrics.} We conduct a comprehensive evaluation of the proposed method using three metrics:
% (1) Accuracy (\textbf{ACC}) measures the proportion of correct answers generated by the RAG system under benign conditions, reflecting its fundamental retrieval and generation capabilities.
(1) Adversarial Accuracy (\textbf{ACC}) assesses the system's ability to produce correct answers when subjected to poisoning attacks;
(2) Attack Success Rate (\textbf{ASR}) represents the proportion of incorrect answers generated due to poisoning attacks. We use GPT-4o to determine whether the LLM’s responses are correct or influenced by poisoning, using prompts are provided in the Section~\ref{sec:prompt};
(3) Provable Average Deviation (\textbf{PAD}) 
% certifies an upper bound on semantic shift in the aggregated retrieval, quantifying residual poisoning and defense effectiveness; lower PAD means better mitigation and preserved semantics, higher indicates greater corruption.
provides a certified upper bound on semantic shifts in the aggregated retrieval representation.
% It measures the residual impact of poisoning attacks and the effectiveness of defense. 
Lower PAD values indicate successful mitigation and preserved 
semantics, while higher values signal greater semantic corruption, making PAD a key metric for attack strength and robustness.

\begin{figure*}[h]
    \centering
    \subfloat[Retrieved Passages Top-$K$]{\includegraphics[width=0.31\linewidth]{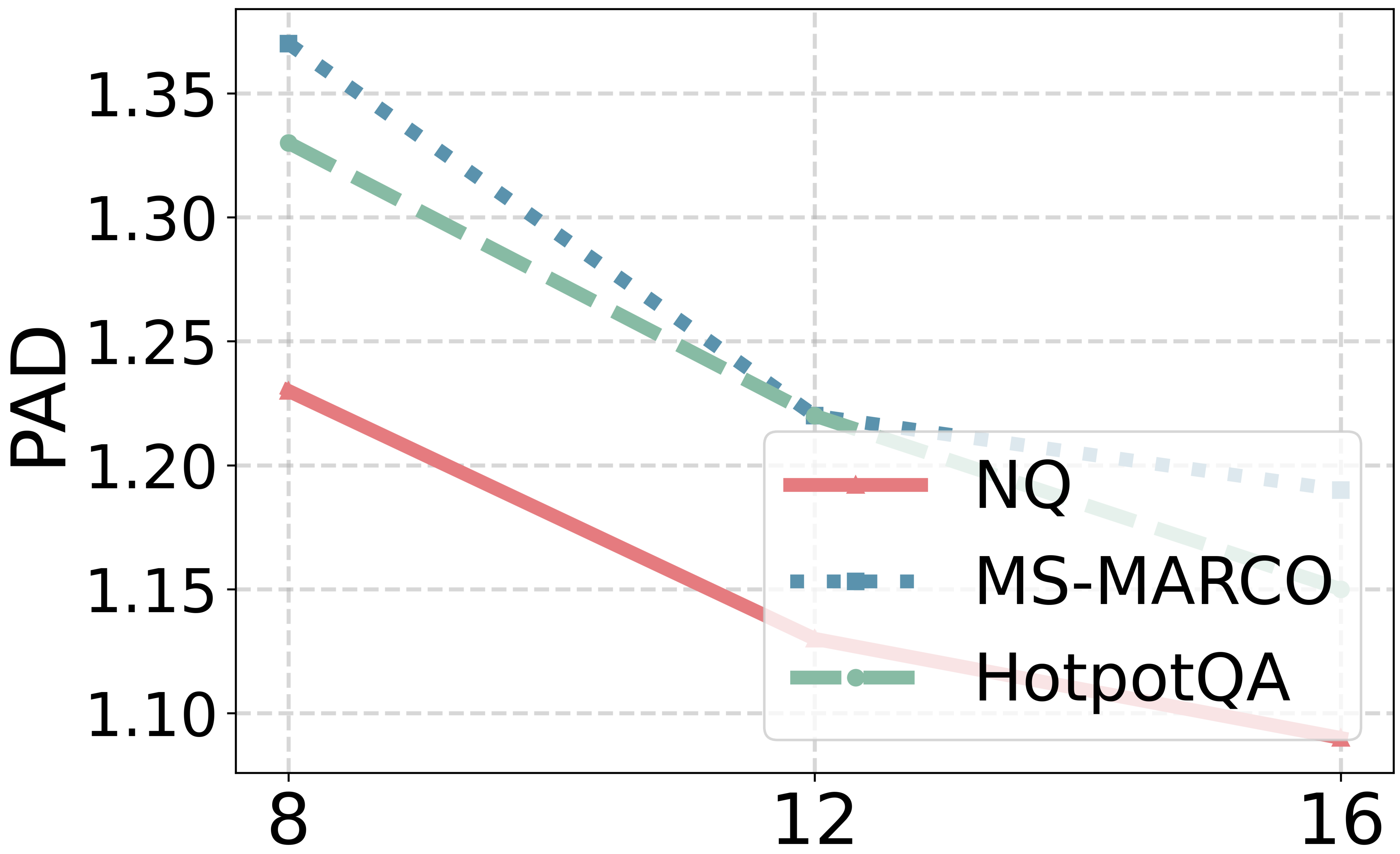}\label{fig:Top-k}}
    \hfill
    \subfloat[Corruption Size $\varepsilon$]{\includegraphics[width=0.31\linewidth]{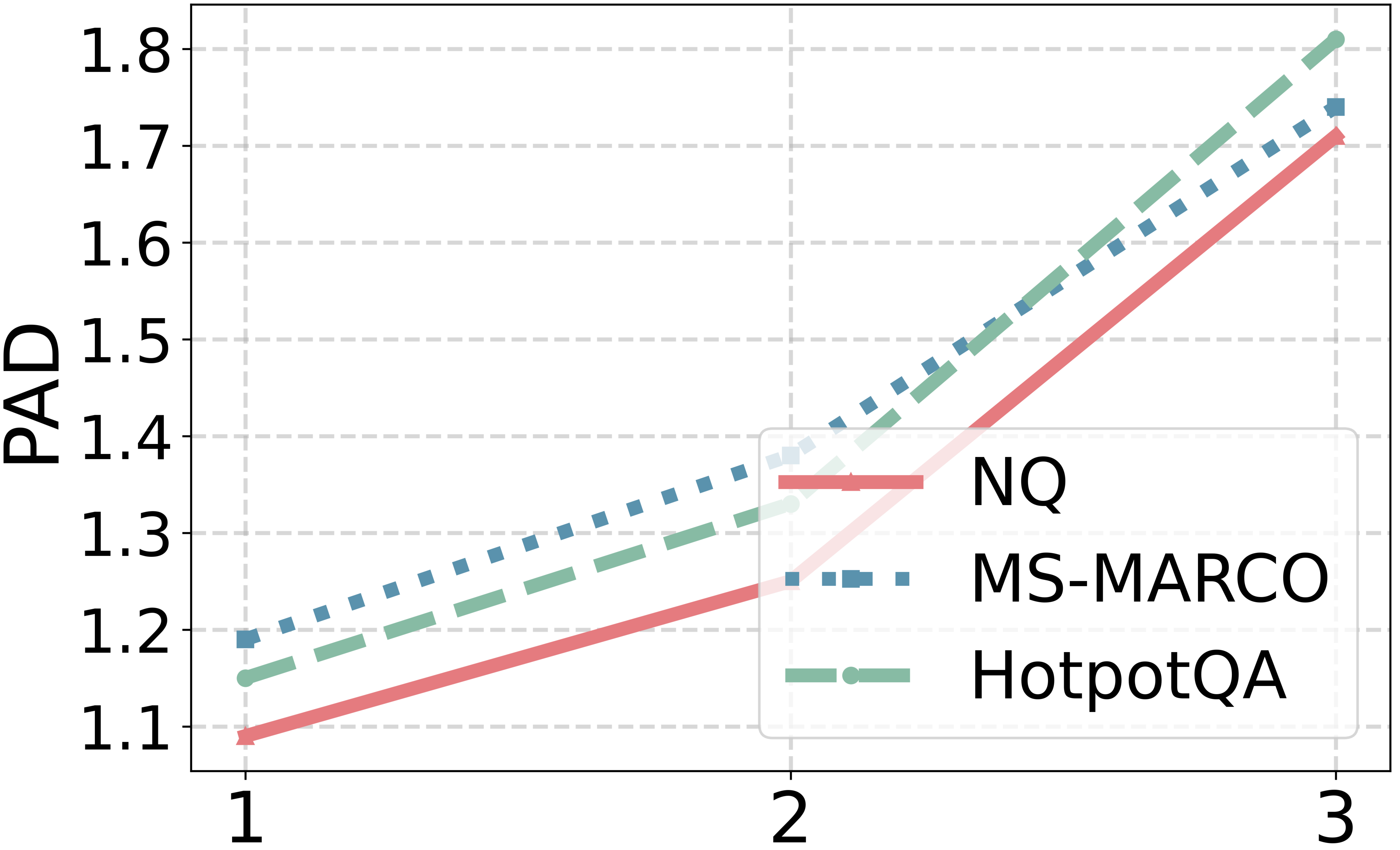}\label{fig:corrupt_size}}
    \hfill
    \subfloat[Subset Size $n$]{\includegraphics[width=0.31\linewidth]{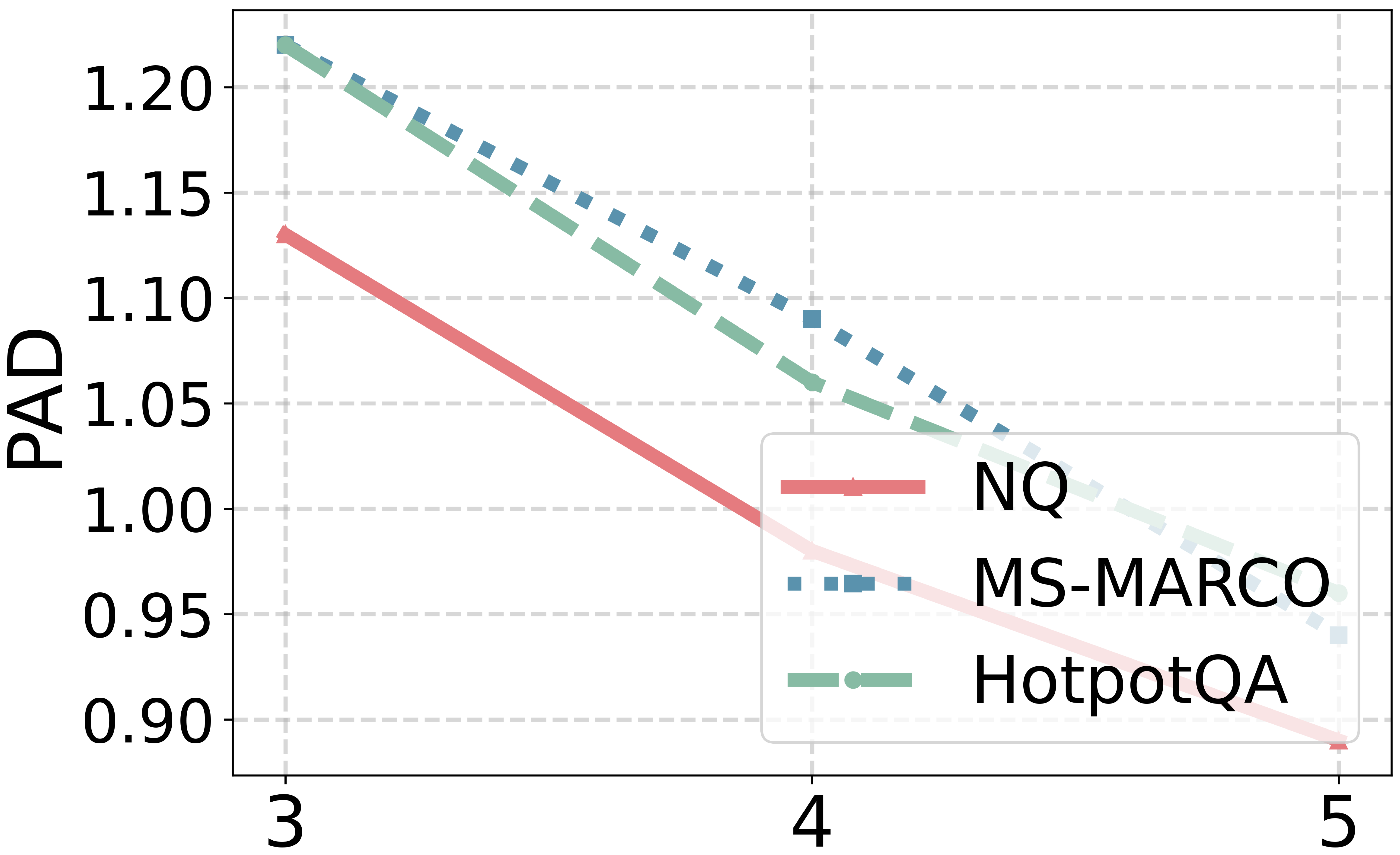}\label{fig:com_size}}
    \caption{Impact of Retrieval Passages, Corruption Size and Subset Size on PAD scores for Mistral-7B.}
    \label{fig:parameter}
\end{figure*}

\subsection{Overall Results}
\textbf{Main results of PRA-RAG.} Table~\ref{table: main_results} summarizes the performance of PRA-RAG under various poisoning attacks across multiple datasets and LLMs. The results demonstrate that PRA-RAG consistently achieves low ASR while maintaining high ACC, indicating strong robustness across all settings. For instance, on the MS-MARCO dataset with Mistral-7B, PRA-RAG attains an ACC of 62\% under PoisonedRAG, with an ASR as low as 1\%. Additionally, PAD effectively quantifies the semantic impact of adversarial retrievals, lower PAD values correspond to lower ASR and higher ACC, supporting its utility as a reliable indicator of both robustness and attack severity.\\
% Additionally, in Section~\ref{sec:methods}, we compare the effects of various embedding extraction methods on the final model performance.
%\noindent
\textbf{PRA-RAG outperforms baselines.} 
Table~\ref{table: baselines} presents a comparative evaluation of PRA-RAG against several representative baselines under poisoning attacks and clean settings. Across all datasets and attack settings, PRA-RAG consistently achieves the best overall performance, attaining higher ACC while significantly lowering the ASR. For instance, on the MSMARCO dataset with 20\% poisoned retrievals, PRA-RAG maintains an ACC of 68\% and suppresses ASR to just 1\% with Vicuna-7B, outperforming RobustRAG and InstructRAG. Meanwhile, we observe that in the absence of poisoning attacks, the $\text{ACC}_{\text{clean}}$ of our method remains largely unaffected. These results demonstrate that PRA-RAG not only offers stronger defense against poisoning attacks on retrieved texts but also preserves utility under normal conditions. Overall, existing methods either lack formal certification, trade usability for safety (over-refusal or over-filtering), or incur high cost with tight coupling to the LLM. PRA-RAG models retrieval poisoning as a set perturbation and performs robust aggregation on \emph{combination embeddings} by selecting a majority-covering MEB center, yielding a certified PAD bound. Under the common regime, it consistently lowers ASR with low overhead, maintains competitive ACC, and withstands semantically-near poisoning.

\begin{table*}[t]
    \centering
    \footnotesize
    \caption{Performance of PRA-RAG and baselines under PoisonedRAG attacks (ACC, ASR) and clean settings ($\text{ACC}_{\text{clean}}$). Results are based on the same experimental setup here for consistency and may differ from the original reports (Top-$K=8$, poisoning rate = 20\%). Best results are in \textbf{bold}.}
    \renewcommand{\arraystretch}{1.2}
    % \resizebox{1.0\linewidth}{!}{
    \begin{tabular*}{\linewidth}{@{\extracolsep{\fill}}ll|ccccccccc}
    \toprule
    \multirow{2}{*}{\textbf{Models}} & \multirow{2}{*}{\textbf{Defense}} 
    & \multicolumn{3}{c}{\textbf{NQ}} 
    & \multicolumn{3}{c}{\textbf{MS-MARCO}}  
    & \multicolumn{3}{c}{\textbf{HotpotQA}}\\
    \cmidrule(lr){3-5} \cmidrule(lr){6-8} \cmidrule(lr){9-11}
    & & $\text{ACC}_{\text{clean}}$ & ACC & ASR 
    & $\text{ACC}_{\text{clean}}$ & ACC & ASR 
    & $\text{ACC}_{\text{clean}}$ & ACC & ASR \\ 
    \midrule
    \multirow{8}{*}{Mistral-7B}
    & Vanilla RAG 
      & 0.61 & 0.26 & 0.46 & 0.58 & 0.33 & 0.18 & 0.55 & 0.17 & 0.48 \\ 
    & RobustRAG$_{\text{Keyword}}$
      & \textbf{0.67} & 0.66 & 0.09 & 0.61 & \textbf{0.66} & 0.12 & \textbf{0.57} & \textbf{0.56} & 0.24 \\ 
    & InstructRAG$_{\text{ICL}}$
      & 0.38 & 0.26 & 0.13 & 0.48 & 0.23 & 0.12 & 0.30 & 0.22 & 0.26 \\ 
    & ASTUTE RAG
      & 0.48 & 0.47 & 0.05 & 0.58 & 0.48 & 0.05 & 0.35 & 0.35 & 0.09 \\ 
    & TrustRAG
      & 0.65 & \textbf{0.62} & 0.03 & \textbf{0.67} & 0.64 & 0.08 & 0.50 & 0.45 & 0.10 \\ 
    & RAGForensics
      & 0.41 & 0.41 & 0.07 & 0.62 & 0.65 & 0.03 & 0.45 & 0.46 & 0.10 \\ 
    & PRA-RAG$_\text{cat}$
      & 0.57 & 0.47 & 0.09 & 0.59 & 0.59 & 0.04 & 0.50 & 0.41 & 0.11 \\ 
    & PRA-RAG
      & 0.60 & \textbf{0.62} & \textbf{0.01} & 0.57 & 0.63 & \textbf{0.01} & 0.49 & 0.37 & \textbf{0.03} \\
    \midrule
    \multirow{8}{*}{Vicuna-7B}
    & Vanilla RAG 
      & \textbf{0.67} & 0.31 & 0.45 & 0.73 & 0.35 & 0.33 & \textbf{0.55} & 0.23 & 0.43 \\ 
    & RobustRAG$_{\text{Keyword}}$
      & 0.48 & 0.40 & 0.08 & 0.59 & 0.40 & 0.07 & 0.44 & 0.31 & 0.10 \\ 
    & InstructRAG$_{\text{ICL}}$
      & 0.33 & 0.10 & 0.11 & 0.54 & 0.33 & 0.20 & 0.31 & 0.09 & 0.11 \\ 
    & ASTUTE RAG
      & 0.50 & 0.44 & 0.19 & \textbf{0.74} & 0.63 & 0.14 & 0.34 & 0.26 & 0.30 \\ 
    & TrustRAG
      & 0.51 & 0.49 & 0.06 & 0.62 & 0.55 & 0.06 & 0.36 & 0.30 & 0.07 \\ 
    & RAGForensics
      & 0.60 & \textbf{0.58} & 0.03 & 0.68 & \textbf{0.71} & 0.04 & 0.51 & \textbf{0.51} & 0.06 \\ 
    & PRA-RAG$_\text{cat}$
      & 0.54 & 0.53 & 0.04 & 0.65 & \textbf{0.71} & 0.02 & 0.49 & 0.37 & 0.13 \\ 
    & PRA-RAG
      & 0.53 & 0.49 & \textbf{0.01} & 0.63 & 0.68 & \textbf{0.01} & 0.50 & 0.49 & \textbf{0.00} \\
    \bottomrule
    \end{tabular*}
    \label{table: baselines}
\end{table*}

\subsection{Ablation Study}
\textbf{Impact of retrieved passages Top-$K$, corruption size $\varepsilon$ and subset size $n$.}
The number of retrieved texts Top-$K$ and poisoned texts  $\varepsilon$ significantly affect both attack effectiveness and defense performance. Figure~\ref{fig:Top-k} illustrates the effect of varying the number of retrieved texts (Top-$K = 8, 12, 16$) on the PAD score of PRA-RAG, with $\varepsilon = 1$ fixed. As Top-$K$ increases, the PAD score gradually declines, indicating enhanced robustness due to increased retrieval diversity. Figure~\ref{fig:corrupt_size} presents the opposite setting: with the number of retrieved texts fixed at Top-$K = 12$, the PAD score increases as the number of poisoned texts rises from $\varepsilon = 1$ to $\varepsilon = 3$. This trend indicates that injecting more poisoned texts into the retrieval results strengthens the attack, thereby diminishing the robustness of PRA-RAG. Figure~\ref{fig:com_size} illustrates the impact of different subset sizes on the performance of PRA-RAG under the setting where the number of retrieved texts is Top-$K = 12$ and the number of poisoned texts is $\varepsilon = 1$. We evaluate cases where each combination contains 3, 4, or 5 texts. The results show a clear downward trend in PAD scores as the number of texts per subset increases. This indicates that incorporating more clean texts within each combination helps dilute the influence of poisoned content.\\ 
% Further experimental data and analyses can be found in Section~\ref{sec:parameters}.\\
\textbf{Comparison of aggregation strategies.} We compare two aggregation methods for the selected subset: a baseline that concatenates texts as input to the LLM (PRA-RAG$_\text{cat}$), and our method (PRA-RAG), which computes a weighted average of embeddings. As shown in Table~\ref{table: baselines}, our approach achieves higher ACC and significantly reduces ASR by better mitigating the impact of poisoned texts. 
% In Section~\ref{sec:dis}, we further compare and analyze the impact of different distance metrics used to compute $d$ on the performance of our method.
% In our framework, we introduce a robustness-aware selection strategy to identify the safest subset from multiple combinations of retrieved texts. For the selected subset, we compare two aggregation methods: a baseline that directly concatenates the selected texts as input to the LLM (PRA-RAG$_\text{cat}$), and our proposed approach (PRA-RAG), which computes a weighted average of the embedding vectors to produce the final aggregated retrieval representation. As shown in Table~\ref{table: baselines}, the weighted averaging method not only achieves higher ACC but also significantly reduces ASR by effectively mitigating the influence of poisoned text.\\

\subsection{Efficiency}

In real-world applications, response time is critical to the user experience of RAG systems. However, most existing defense strategies often incur varying levels of inference overhead. 
In our experiments, we evaluated the average response latency over 100 query examples to compare different defense methods. As shown in Table~\ref{table:latency}, our proposed PRA-RAG achieves the lowest inference latency, demonstrating its practical efficiency while preserving strong robustness. In our framework, both the number of retrieved texts (Top-$K$) and the subset size have a significant impact on system efficiency. Therefore, we conduct a more in-depth analysis and empirical evaluation of the computational overhead in Section~\ref{sec:cost}. \textcolor{blue}{Additionally, Section~\ref{sec:Monte Carlo} presents the computational overhead and analysis under Monte Carlo sampling.}

\renewcommand{\arraystretch}{1.8} % 默认是 1.0，调大数值可增大行距
\begin{table}[t]
\centering
\footnotesize
\caption{The table presents the average response latency of different RAG methods across datasets.}
\resizebox{\linewidth}{!}{\huge
\begin{tabular}{lcccc}
\toprule
\textbf{Dataset} & \textbf{ASTUTE RAG} & \textbf{InstructRAG$_{\text{ICL}}$} & \textbf{RobustRAG$_{\text{Keyword}}$} & \textbf{PRA-RAG} \\
\midrule
NQ        & 13.68s & 6.89s & 27.01s &\textbf{ 5.98s} \\ 
MS-MARCO  & 13.84s & 7.15s & 26.38s & \textbf{6.20s} \\ 
HotpotQA  & 16.78s & 7.01s & 25.87s & \textbf{6.17s} \\ 
\bottomrule
\end{tabular}}
\label{table:latency}
\end{table}

\section{Conclusion}

This paper presents PRA-RAG, a provably robust aggregation method for defending against retrieval-level poisoning attacks in RAG systems. We theoretically derive an upper bound on the semantic deviation introduced by poisoned retrievals and propose Provable Average Deviation (PAD) as a unified metric for evaluating both robustness and attack strength. Unlike traditional defenses that depend on model outputs, PRA-RAG ensures semantic stability directly in the embedding space, providing a principled, model-agnostic robustness guarantee. Extensive experiments show that PRA-RAG significantly enhances robustness against poisoning attacks across diverse datasets and language models, consistently lowering attack success rates while preserving high answer accuracy. 
% In future work, we will extend PAD to other retrieval-based tasks, including fact verification and knowledge-grounded dialogue, where robustness is essential.

% In future work, we plan to extend PAD to other retrieval-based tasks such as fact verification and knowledge-grounded dialogue, where retrieval robustness remains essential.

% Moreover, PAD provides a reliable and interpretable signal for quantifying the severity of semantic corruption, making it a valuable tool for analyzing system resilience. 

\noindent\textbf{Limitations.}

\noindent Our work has the following limitations:
\begin{itemize}
  \item This work selects the most robust subset of the retrieved texts as the final input to the LLM for answer generation. Although the selected subset may not always produce the optimal response from the LLM, as shown in our experiments, the impact on model utility remains limited while significantly enhancing robustness against poisoning attacks.
  \item Our approach enhances robustness by constructing subsets, which increases computational and response-time overhead. However, experiments show that as the number and size of subsets grow, the defense against poisoning attacks significantly improves. In practice, users can flexibly balance efficiency and security based on their needs. 
  \item The effectiveness of our approach relies on the number of poisoned texts remaining below a certain threshold. Although an attacker could potentially bypass the defense through large-scale poisoning, this incurs substantial costs and significantly reduces the quality of the context, making it difficult for even humans to provide correct answers.
\end{itemize}

\noindent\textbf{Ethics Statement}

\noindent The goal of this work is to defend against retrieval-based poisoning attacks in RAG systems. All data used in this study is publicly available, ensuring no additional privacy concerns. The source code and software will be released as open-source. While this openness may expose the system to adaptive attacks, our approach can be further strengthened by incorporating additional internal and external information. 
% Overall, we believe our method contributes to advancing the secure deployment of RAG systems.

\noindent\textbf{Acknowledgment.}

\noindent Xue Tan, Yi Zheng, Chang Huo, Yunruo Zhang, Yu Liu, Hao Luan, Zhuyang Yu, and Ping Chen were funded in part by the National Key R\&D Program of China 2023YFB3107404.

\bibliography{ref}

\appendix

\section{Theoretical Analysis and Proofs}
\label{sec:proof}
\subsection{Proof of Theorem 1}
\label{sec:proof1}
Consider two balls $\mathcal{B}(\mathcal{F}(\mathcal{X'}), r^*(\mathcal{X'}))$ and $\mathcal{B}(\mathcal{F}(\mathcal{X}), R)$. According to the definitions of $R$ and $r^*(\mathcal{X'})$, these two balls must share at least one common element, denoted as $w^*$. Since the defined distance based on vector angles satisfies the triangle inequality, we have:
\begin{equation}
  \label{eq:triangle}
  \begin{aligned}
    d(\mathcal{F}(\mathcal{X}), \mathcal{F}(\mathcal{X'})) 
    &\le d(\mathcal{F}(\mathcal{X}), w^*) + d(\mathcal{F}(\mathcal{X'}), w^*) \\
    &\le R + r^*(\mathcal{X'}).
  \end{aligned}
\end{equation}
Since the ball $\mathcal{B}(\mathcal{F}(\mathcal{X}), R)$ contains at least a majority of elements in $\mathcal{V}_{\mathcal{X'},n}$, the smallest radius $r^*(\mathcal{X'})$ of a ball that contains a majority of elements in $\mathcal{V}_{\mathcal{X'},n}$ must not be greater than $R$, i.e., $r^*(\mathcal{X'}) \le R$. Therefore,
\begin{equation}
d(\mathcal{F}(\mathcal{X}), \mathcal{F}(\mathcal{X'})) \le 2R.
\end{equation}

\subsection{Proof of Theorem 2}
\label{sec:proof2} 
Since the radius of the $\beta$-MEB differs from the optimal radius by a factor $\beta$, i.e., the radius of the $\beta$-MEB over $\mathcal{V}_{\mathcal{X},n}$ is $\beta r^*(\mathcal{X}')$, by the same reasoning as in Theorem 1, we have:
\begin{equation}
  \label{eq:chain-ineq}
  \begin{aligned}
    d(\widehat{\mathcal{F}}(\mathcal{X}), \widehat{\mathcal{F}}(\mathcal{X}')) 
    &\le d(\widehat{\mathcal{F}}(\mathcal{X}), w^*) + d(\widehat{\mathcal{F}}(\mathcal{X}'), w^*) \\
    &\le \widehat{R} + \beta r^*(\mathcal{X}') 
    \le (1 + \beta)\widehat{R}.
  \end{aligned}
\end{equation}

\subsection{Scalability and Approximation via Monte-Carlo Sampling}
\label{sec:appendix_sampling_proof}

While the theoretical formulation in Section~\ref{sec:the} relies on the full set of $L = \binom{K}{n}$ combinations, the computational cost grows combinatorially as $K$ and $n$ increase. To address this bottleneck, we introduce a Monte-Carlo sampling strategy: instead of enumerating all $L$ combinations, we uniformly sample $m \ll L$ subsets to estimate the robust aggregation. Below, we establish a strict theoretical lower bound on the sampling approximation and provide an empirically supported upper bound, demonstrating how this strategy effectively approximates the optimal robustness radius.

\begin{proposition}
\textbf{Lower Bound of Monte-Carlo Approximation.}
\label{prop:sampling_lower}
Let $\mathcal{V}_{\mathcal{X}, n}$ be the full set of embedding vectors generated from $\binom{K}{n}$ combinations, and let $R$ be the radius of the Minimum Enclosing Ball (MEB) covering the majority ($>50\%$) of $\mathcal{V}_{\mathcal{X}, n}$. Let $\mathcal{S} \subset \mathcal{V}_{\mathcal{X}, n}$ be a uniformly random sample of size $m$, and let $R^*$ be the radius of the MEB covering the majority of $\mathcal{S}$. For any $\tau \in (0,\,0.5)$ and failure probability $\eta > 0$, if the sample size satisfies $m \ge \frac{1}{2\tau^2}\ln\frac{2}{\eta}$ and the majority coverage fraction $q > 0.5 + \tau$, then with probability at least $1-\eta$:
\begin{equation}
    R^* \le R.
\end{equation}
\end{proposition}

\begin{proof}
\textit{Proof of Lower Bound.}
% \textit{Coverage Concentration.}
For any fixed ball $\mathcal{B}$ in the embedding space, let $q$ denote the fraction of $\mathcal{V}_{\mathcal{X},n}$ within $\mathcal{B}$, and $\hat{q}$ the fraction in $\mathcal{S}$. By Hoeffding's Inequality~\cite{hoeffding1963probability}:
\begin{equation}
    \mathbb{P}(|\hat{q} - q| \ge \tau) \le 2\exp(-2m\tau^2).
\end{equation}
Under the condition $m \ge \frac{1}{2\tau^2}\ln\frac{2}{\eta}$, we have $2\exp(-2m\tau^2) \le \eta$. For instance, with $m = 200$ and $\tau = 0.1$, the failure probability is below $0.04$.

\textit{Lower Bound Derivation.}
Let $\mathcal{B}^*$ be the optimal ball with radius $R$ covering a fraction $q > 0.5 + \tau$ of $\mathcal{V}_{\mathcal{X},n}$. By the coverage concentration above, $\mathcal{B}^*$ covers a fraction at least $q - \tau > 0.5$ of $\mathcal{S}$ with high probability. Thus $R$ is a feasible majority-covering radius for $\mathcal{S}$. Since $R^*$ is the minimum such radius over $\mathcal{S}$:
\begin{equation}
    R^* \le R.
\end{equation}
\end{proof}

\begin{remark}
\textbf{Tightness of the Approximation.}
\label{remark:upper_bound}
The lower bound $R^* \le R$ guarantees that sampling does not overestimate the robustness radius. A complementary strict upper bound of the form $R \le (1+\tau)\,R^*$ would require guaranteeing that the random sample $\mathcal{S}$ captures the extremal boundary points that determine the MEB of the full set. Since these boundary points may constitute a vanishing fraction of $\mathcal{V}_{\mathcal{X},n}$, such a deterministic worst-case guarantee cannot be established from uniform sampling alone without additional structural assumptions.
 
However, in our setting, the Geometric Separability Assumption (Section~\ref{app:separability}) implies that clean combination embeddings form a concentrated cluster, which facilitates effective geometric approximation from moderate-sized samples. Empirically, across all experimental configurations in Section~\ref{sec:Monte Carlo}, the difference between $R$ and $R^*$ is consistently small. Tables~\ref{table: topk} and~\ref{table: n} show that the PAD values under Monte-Carlo sampling closely match those under full enumeration, with deviations typically within $3\%$, confirming that the approximation introduces no significant performance degradation. Furthermore, ACC and ASR remain virtually unchanged, and the overall trends (e.g., decreasing PAD with increasing Top-$K$ and $n$) are fully preserved under sampling.
\end{remark}
 
Combining the strict lower bound with the empirically validated tightness, Monte-Carlo sampling with $m \approx 200$ yields an effective approximation of the true robustness metric, decoupling the computational complexity from the combinatorial growth of $\binom{K}{n}$ and reducing it to a constant factor determined by the sample size $m$.

\subsection{\textcolor{blue}{Guarantee of Clean Majority in the Certified Ball}}
\label{lemma1}
\textcolor{blue}{\textbf{Lemma 1.} Let $\mathcal{V}$ be the set of all $L = \binom{K}{n}$ subset embeddings derived from the retrieved documents. Let $\mathcal{C} \subset \mathcal{V}$ denote the set of clean embeddings and $\mathcal{P} \subset \mathcal{V}$ denote the set of poisoned embeddings, such that $|\mathcal{P}| \le L_{adv} = \binom{K}{n} - \binom{K-\varepsilon}{n}$.
If the certified radius $\hat{R}$ is determined by the distance value at the $k$-th index (0-based) of the sorted distance vector to the geometric center, with $k = \lfloor L/2 \rfloor + L_{adv}$, then the ball $\mathcal{B}(z, \hat{R})$ is guaranteed to strictly contain a majority of clean embeddings, i.e., $|\mathcal{B}(z, \hat{R}) \cap \mathcal{C}| \ge \lfloor L/2 \rfloor + 1$.
}\\
\textcolor{blue}{
\noindent \textbf{\textit{Proof.}} Let $\mathcal{S}$ denote the set of embeddings enclosed by the ball $\mathcal{B}(z, \hat{R})$. Since the radius $\hat{R}$ corresponds to the distance at the $k$-th index (where index 0 represents the center itself), the set $\mathcal{S}$ includes the first $k+1$ smallest distance embeddings. Thus, by definition, $|\mathcal{S}| = k + 1$.
The set $\mathcal{S}$ is composed of disjoint clean and poisoned subsets: $\mathcal{S} = (\mathcal{S} \cap \mathcal{C}) \cup (\mathcal{S} \cap \mathcal{P})$.
We seek a lower bound on the number of clean samples $|\mathcal{S} \cap \mathcal{C}|$.
According to set theory:
\begin{equation}
    |\mathcal{S} \cap \mathcal{C}| = |\mathcal{S}| - |\mathcal{S} \cap \mathcal{P}|
\end{equation}
In the worst-case adversarial scenario, the adversary optimizes the poisoned embeddings to be geometrically closest to the center to occupy the top ranks. However, the total number of poisoned embeddings is strictly bounded by $L_{adv}$. Thus, for any selected subset, $|\mathcal{S} \cap \mathcal{P}| \le L_{adv}$ always holds.
Substituting the cardinality $|\mathcal{S}| = k + 1$ and the value of $k$:
\begin{align}
    |\mathcal{S} \cap \mathcal{C}| &\ge (k + 1) - L_{adv} \notag \\
    &= (\lfloor L/2 \rfloor + L_{adv} + 1) - L_{adv} \notag \\
    &= \lfloor L/2 \rfloor + 1
\end{align}
Therefore, the ball $\mathcal{B}(z, \hat{R})$ necessarily contains at least $\lfloor L/2 \rfloor + 1$ clean embeddings. Since $L = |\mathcal{V}|$, this count represents a strict majority of the total universe of combinations.
}

\subsection{\textcolor{blue}{Geometric Separability Assumption}}
\label{app:separability}

\textcolor{blue}{We assume a mild geometric separability condition to connect majority coverage
over combination embeddings to clean subset selection.
Specifically, there exists a center $z^*$ and a radius $R_c$
such that a strict majority of fully clean combination embeddings
lie within the ball $\mathcal{B}(z^*, R_c)$,
while any combination containing at least one poisoned passage lies outside
$\mathcal{B}(z^*, R_c + \Delta)$ for some $\Delta > 0$.}

\textcolor{blue}{The separability condition is necessary to ensure that a ball containing a strict majority of clean embeddings is geometrically anchored to the clean subset, thereby enabling the selection of clean contexts through geometric aggregation. When the margin $\Delta = 0$, clean and poisoned combination embeddings become geometrically indistinguishable. In such cases, no geometry-based majority certification method can guarantee the selection of a fully clean subset, as the observed information no longer provides distinguishable features.}

\section{Details}
\subsection{Details of Experiment Setup}
\label{sec:setup}
\textbf{Dataset.}
Detailed statistics are shown in Table~\ref{table:datasets}.
\begin{table}[htbp]
\centering
\footnotesize
\caption{Statistics of datasets.}
\resizebox{\linewidth}{!}{ % This will resize the table to fit the text width
\begin{tabular}{lcc}
\toprule
\textbf{Datasets} & \textbf{Texts in knowledge database} & \textbf{Questions} \\
\midrule
NQ~\cite{kwiatkowski2019natural} & 2,681,468 & 3,452  \\ 
MS-MARCO~\cite{bajaj2016ms}  & 5,233,329 & 7,405  \\ 
HotpotQA~\cite{yang2018hotpotqa}  & 8,841,823 & 6,980  \\ 
\bottomrule
\end{tabular}}
\label{table:datasets}
\end{table}\\
\textbf{Default Setting.} Unless otherwise specified, our experiments use the following default settings: the Contriever~\cite{izacard2021unsupervised} retriever is employed to select the Top-$K=8$ texts most similar to the query from the MS-MARCO dataset, with poisoned texts generated according to the PoisonedRAG~\cite{zou2024poisonedrag} method constituting 20\% (rounded down) of the retrieval results. Subsequently, multiple subsets are sampled from the retrieved texts, each containing $n=3$ texts. Open-source LLMs can directly obtain embeddings through their embedding layers and generate the final robust aggregated representation. For black-box models (e.g., GPT-3.5-turbo and GPT-5-mini), we use Llama3-8B in our experiments to extract embeddings and decode the robust aggregation results into text to serve as context. Finally, the Mistral-7B language model generates responses based on these retrieved texts. Each result is averaged over 10 runs under a consistent setup.\\
\textbf{Implementation and Resources.} 
% We report the mean performance over 10 independent runs and assess the statistical significance of improvements or degradations using the Wilcoxon signed-rank test at a significance level of $p \le 0.05$. 
We conduct experiments on a server with 64 AMD EPYC 9654 CPUs at 2.40–3.70 GHz, 512 GB of DDR5 RAM, and four NVIDIA RTX A100 GPUs, each with 80 GB HBM2e memory.

\subsection{Prompt}
\label{sec:prompt}

\begin{figure}[H]
    %\vspace{-7mm}
    \centering
    \includegraphics[width=0.9\linewidth]{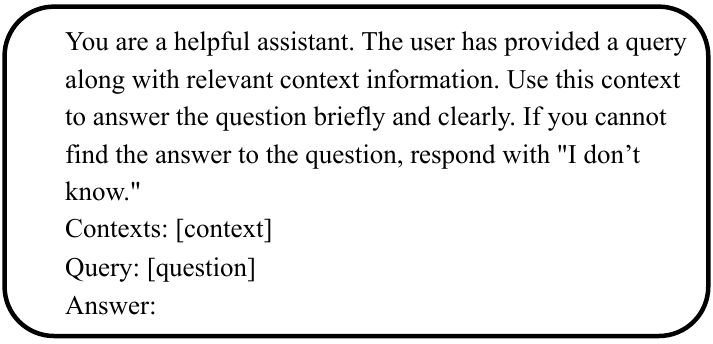}
    \caption{The prompt used in RAG to make an LLM generate an answer based on the retrieved texts.}
    \label{fig:prompt1}
\end{figure}

\begin{figure}[H]
    %\vspace{-7mm}
    \centering
    \includegraphics[width=0.9\linewidth]{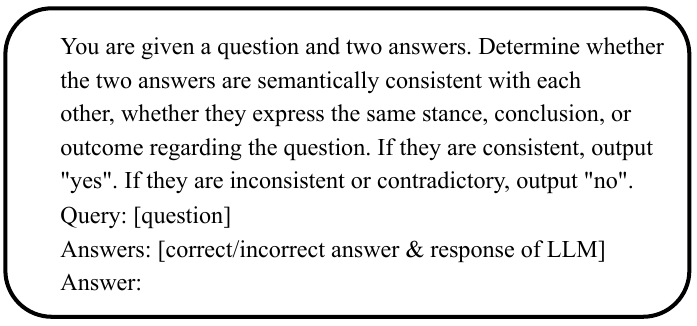}
    \caption{The prompt used in RAG to determine whether the LLM’s response is correct or corresponds to the attacker’s target response.}
    \label{fig:prompt2}
\end{figure}

\subsection{Poisoned Texts Generation}
\label{sec:poi}
To comprehensively evaluate the robustness of our approach against different poisoned text generation strategies, we adopt three representative strategies, including PoisonedRAG~\cite{zou2024poisonedrag}, Adversarial Decoding~\cite{zhang2025adversarialdecoding}, Denial-Of-Service Attack~\cite{shafran2024machine}, \textcolor{blue}{and CorruptRAG~\cite{zhang2025practical}}. In PoisonedRAG~\cite{zou2024poisonedrag} method, the attacker begins by selecting a target question and its corresponding incorrect answer, then crafts poisoned texts to satisfy two key requirements: (1) being retrievable by the retriever and (2) successfully misleading the language model into producing the incorrect answer. Adversarial Decoding~\cite{zhang2025adversarialdecoding} is a token-level beam search framework that integrates continuous, task-specific scorers into standard decoding, enabling the generation of fluent, low-detectability texts that jointly optimize retrieval similarity and adversarial generation objectives. Denial-Of-Service Attack~\cite{shafran2024machine} inserts a single blocker document consisting of a retrieval component to ensure retrieval and a jamming component to trigger refusal responses, crafted through instruction injection, oracle generation, or black-box optimization. \textcolor{blue}{CorruptRAG~\cite{zhang2025practical} can efficiently mislead the model into generating targeted false content by injecting only a single poisoned text that simultaneously includes the target query, which ensures it is preferentially retrieved, and an adversarial prompting template, which exploits the LLM’s generation bias to label the correct answer as outdated and to fabricate a supposedly latest but incorrect answer.} In our experiments, we adopt the default parameters of these methods to generate the corresponding poisoned texts, ensuring effective poisoning.

\label{sec:cluter}
\begin{figure}[H]
    \centering
    \subfloat[Mistral-7B]{\includegraphics[width=0.5\linewidth]{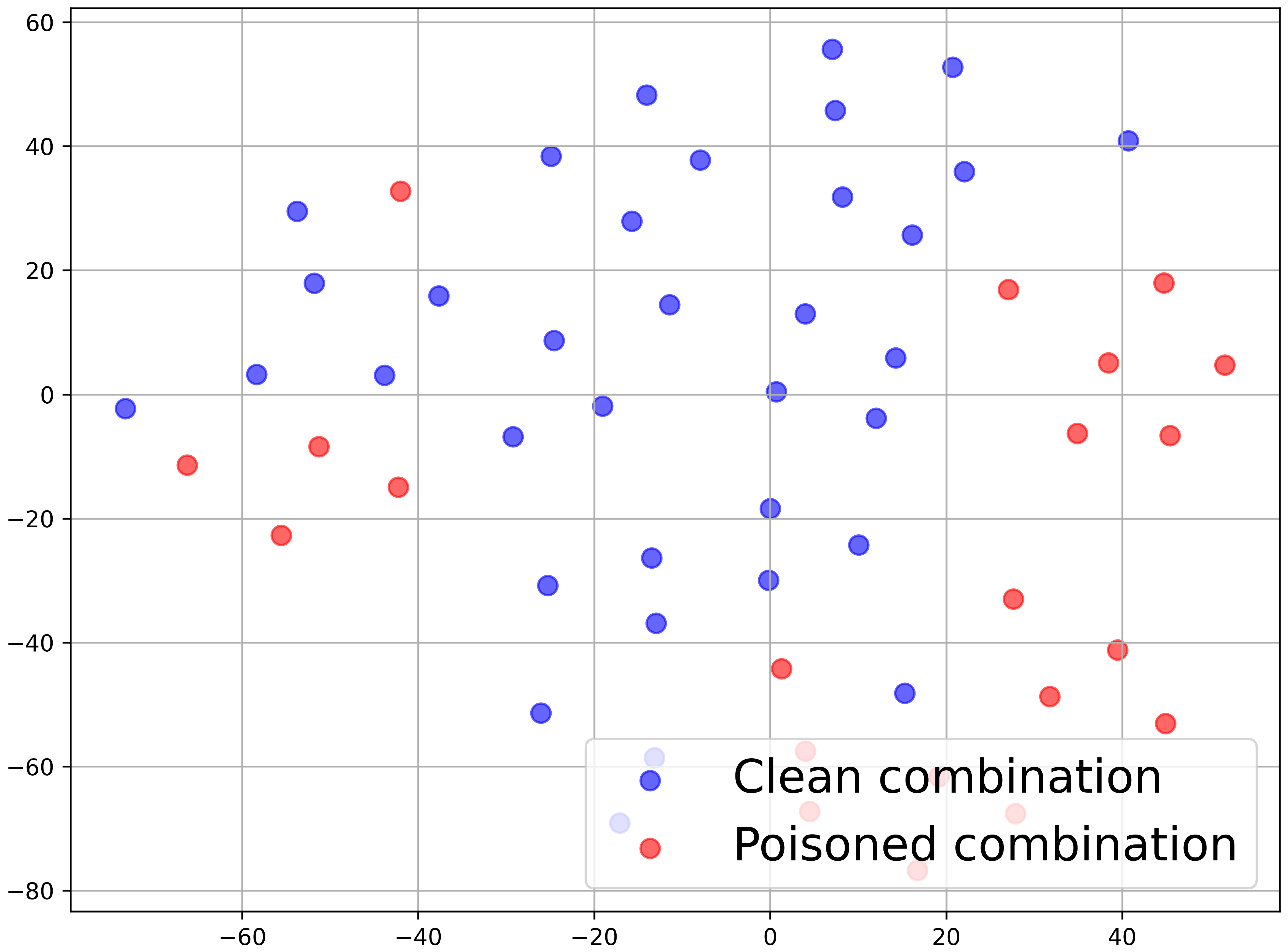}\label{fig:mistral}}
    \hfill
    \subfloat[Vicuna-7B]{\includegraphics[width=0.5\linewidth]{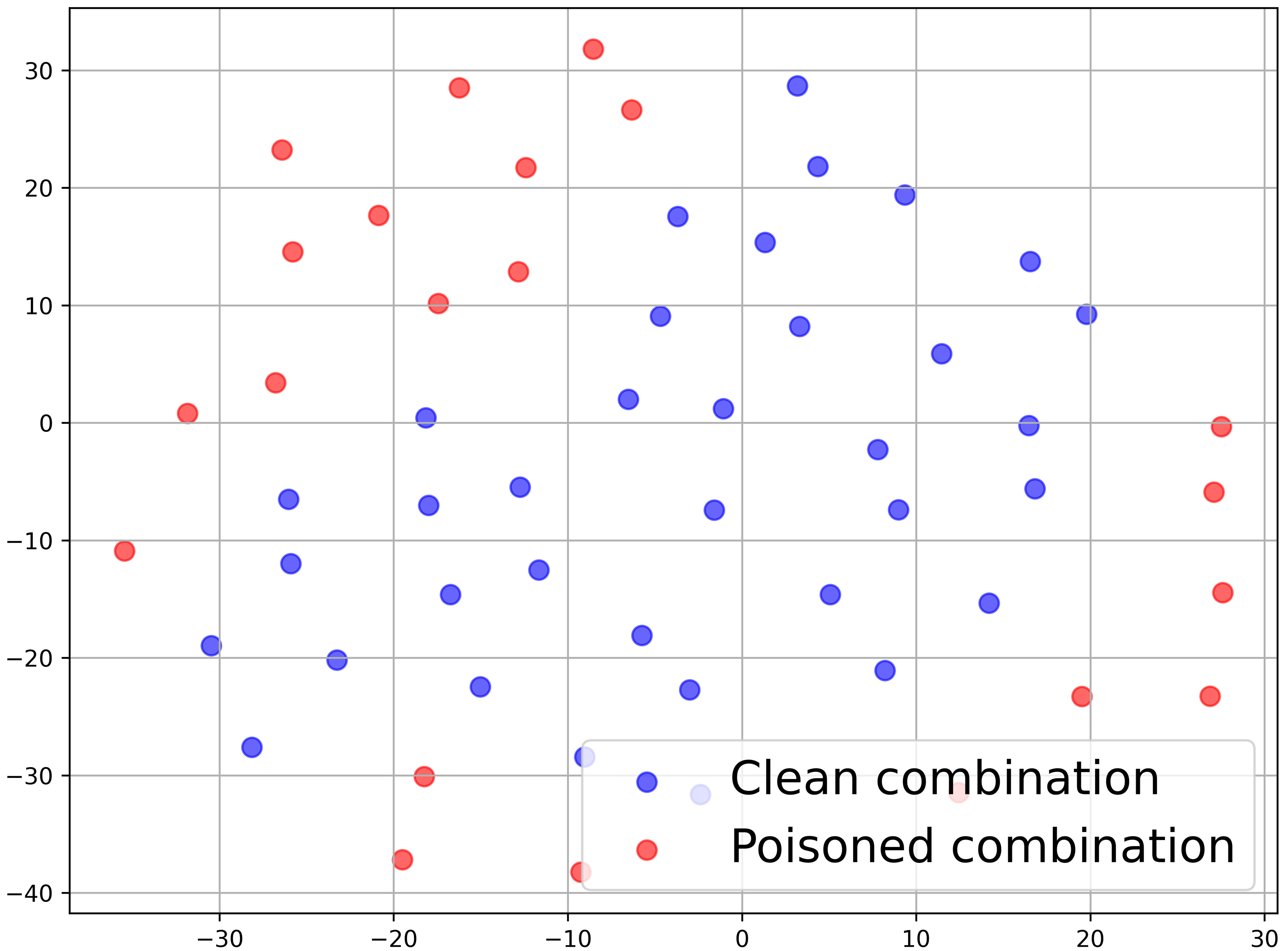}\label{fig:vicuna}}
    \caption{Embedding Distribution of Retrieved Text Combinations (Red: With Poisoned Texts; Blue: Without Poisoned Texts)}
    \label{fig:cluster}
\end{figure}

\section{Additional Experimental Results}
\label{sec:appendix}
\subsection{Distribution of Different Combinations}
\label{sec:distribution}

Figure~\ref{fig:cluster} illustrates the distributional differences between combinations containing poisoned texts and clean combinations. It is evident that the poisoned texts, intentionally crafted to induce attacker-desired responses, exhibit a significant shift in their embedding vectors compared to clean texts.

\renewcommand{\arraystretch}{1.2} % 默认是 1.0，调大数值可增大行距
\begin{table}[t]
    \centering
    \footnotesize
    \caption{Computational overhead of PRA-RAG under different Top-$K$ retrieval settings, with subset size fixed at 3 and a poisoning rate of 20\%.}
    \resizebox{0.99\linewidth}{!}{%
    \begin{tabular}{lcccc}
        \toprule
        \textbf{Models} & \textbf{Top-$K$} & \textbf{NQ} & \textbf{MS-MARCO} & \textbf{HotpotQA} \\
        \midrule
        \multirow{3}{*}{Mistral-7B} 
        & 8  & 5.98s & 6.20s & 6.17s \\
        & 12 & 18.47s & 24.88s & 22.14s \\
        & 16 & 68.05s & 65.48s & 72.16s \\
        \midrule
        \multirow{3}{*}{Vicuna-7B} 
        & 8  & 3.37s & 2.91s & 3.02s \\
        & 12 & 12.13s & 11.58s & 12.23s \\
        & 16 & 64.93s & 67.13s & 59.25s \\
        \bottomrule
    \end{tabular}}
    \label{table:cost_topk-1}
\end{table}

\renewcommand{\arraystretch}{1.2} % 默认是 1.0，调大数值可增大行距
\begin{table}[t]
    \centering
    \footnotesize
    \caption{Computational overhead of PRA-RAG with Top-$K = 12$ retrieval under a poisoning rate of 10\%, evaluated across different subset sizes.}
    \resizebox{0.99\linewidth}{!}{%
    \begin{tabular}{lcccc}
        \toprule
        \textbf{Model} & \textbf{Subset Size} & \textbf{NQ} & \textbf{MS-MARCO} & \textbf{HotpotQA} \\
        \midrule
        \multirow{3}{*}{Mistral-7B} 
        & 3 & 18.71s & 18.10s & 17.76s \\
        & 4 & 55.25s & 55.30s & 57.89s \\
        & 5 & 125.55s & 131.37s & 124.96s \\
        \midrule
        \multirow{3}{*}{Vicuna-7B} 
        & 3 & 11.88s & 12.16s & 11.84s \\
        & 4 & 47.92s & 46.12s & 52.40s \\
        & 5 & 116.09s & 121.47s & 117.85s \\
        \bottomrule
    \end{tabular}}
    \label{table:cost_subset-1}
\end{table}

\subsection{Analysis of Computational Overhead}
\label{sec:cost}
Table~\ref{table:cost_topk-1} compares the computational overhead under different numbers of retrieved texts, with the subset size $n=3$. The results show that as Top-$K$ increases, the number of combinations grows rapidly, leading to a near-linear increase in computation time. Specifically, when Top-$K$ is set to 8, 12, and 16, the average generation times are 5.98 s, 18.47 s, and 68.05 s, respectively. Table~\ref{table:cost_subset-1} compares the computational overhead under a fixed Top-$K$ retrieval size of 12 with varying subset sizes. The results show that the computation time no longer correlates linearly with the number of combinations, as changes in subset size also affect the overall time required for the final RAG response generation. Both the number of retrieved texts (Top-$K$) and the subset size ($n$) have a significant impact on the response time of our approach. Nevertheless, as shown in Table~\ref{table:latency}, our method still outperforms existing baselines in terms of efficiency under certain settings (e.g., Top-$K=8$, $n=3$), which represent typical real-world scenarios. Although increasing Top-$K$ and $n$ leads to higher computational overhead, it also enhances system security. As shown in Figure~\ref{fig:parameter}, the PAD value decreases as Top-$K$ and $n$ increase, indicating higher security. Overall, users can flexibly balance efficiency and security by selecting appropriate values of Top-$K$ and $n$ based on practical needs.

\begin{table}[t]
    \centering
    \footnotesize
    % \color{blue}
    % \captionsetup{labelfont={color=blue}, textfont={color=blue}}
    \caption{Performance of PRA-RAG using Monte Carlo sampling under varying Top-$K$ (subset size = 3, poisoning rate = 10\%).}
    \resizebox{0.99\linewidth}{!}{\huge
    \begin{tabular}{lcccccccccc}
    \toprule
    \multirow{2}{*}{\textbf{Models}} & \multirow{2}{*}{\textbf{Top-$K$}}
    & \multicolumn{3}{c}{\textbf{NQ}}
    & \multicolumn{3}{c}{\textbf{MS-MARCO}}
    & \multicolumn{3}{c}{\textbf{HotpotQA}} \\
    \cmidrule(lr){3-5} \cmidrule(lr){6-8} \cmidrule(lr){9-11}
    & & ACC & ASR & PAD & ACC & ASR & PAD & ACC & ASR & PAD \\
    \midrule
    \multirow{3}{*}{Mistral-7B}
    & 8
      & 0.70 & 0.00 & 1.28
      & 0.62 & 0.00 & 1.31
      & 0.34 & 0.12 & 1.32 \\ 
    & 12
      & 0.52 & 0.00 & 1.16
      & 0.66 & 0.00 & 1.21
      & 0.44 & 0.02 & 1.18 \\ 
    & 16
      & 0.52 & 0.00 & 1.12
      & 0.68 & 0.00 & 1.20
      & 0.40 & 0.02 & 1.15 \\ 
    \midrule
    \multirow{3}{*}{Vicuna-7B}
    & 8
      & 0.44 & 0.00 & 1.78
      & 0.66 & 0.02 & 1.84
      & 0.42 & 0.04 & 1.86 \\ 
    & 12
      & 0.51 & 0.00 & 1.65
      & 0.76 & 0.00 & 1.72
      & 0.52 & 0.01 & 1.74 \\ 
    & 16
      & 0.42 & 0.00 & 1.64
      & 0.70 & 0.00 & 1.70
      & 0.48 & 0.00 & 1.67 \\ 
    \bottomrule
    \end{tabular}}
    \label{table: topk}
\end{table}

\begin{table}[t]
    \centering
    \footnotesize
    % \color{blue}
    % \captionsetup{labelfont={color=blue}, textfont={color=blue}}
    \caption{Performance of PRA-RAG using Monte Carlo sampling with Top-$K = 12$ retrieval under a poisoning rate of 10\%, evaluated across different subset sizes.}
    \resizebox{0.99\linewidth}{!}{\huge
    \begin{tabular}{lcccccccccc}
    \toprule
    \multirow{2}{*}{\textbf{Models}} & \multirow{2}{*}{\textbf{Subset Size}}
    & \multicolumn{3}{c}{\textbf{NQ}}
    & \multicolumn{3}{c}{\textbf{MS-MARCO}}
    & \multicolumn{3}{c}{\textbf{HotpotQA}} \\
    \cmidrule(lr){3-5} \cmidrule(lr){6-8} \cmidrule(lr){9-11}
    & & ACC & ASR & PAD & ACC & ASR & PAD & ACC & ASR & PAD \\
    \midrule
    \multirow{3}{*}{Mistral-7B}
    & 3
      & 0.52 & 0.00 & 1.16
      & 0.66 & 0.00 & 1.21
      & 0.44 & 0.02 & 1.18 \\ 
    & 4
      & 0.56 & 0.00 & 1.02
      & 0.66 & 0.00 & 1.07
      & 0.44 & 0.02 & 1.07 \\ 
    & 5
      & 0.72 & 0.00 & 0.90
      & 0.82 & 0.00 & 0.97
      & 0.46 & 0.01 & 0.96 \\ 
    \midrule
    \multirow{3}{*}{Vicuna-7B}
    & 3
      & 0.51 & 0.00 & 1.65
      & 0.76 & 0.00 & 1.72
      & 0.52 & 0.01 & 1.74  \\ 
    & 4
      & 0.48 & 0.00 & 1.44
      & 0.66 & 0.02 & 1.51
      & 0.47 & 0.02 & 1.55 \\ 
    & 5
      & 0.47 & 0.00 & 1.30
      & 0.60 & 0.00 & 1.37
      & 0.45 & 0.02 & 1.36 \\ 
    \bottomrule
    \end{tabular}}
    \label{table: n}
\end{table}

\begin{table}[t]
    \centering
    \footnotesize
    \caption{Performance of PRA-RAG using Monte Carlo sampling with Top-$K = 20$ retrieval under a poisoning rate of 20\%, evaluated across different sampling numbers and subset sizes under Mistral-7B.}
    \resizebox{0.99\linewidth}{!}{\huge
    \begin{tabular}{lcccccccccc}
    \toprule
    \multirow{1}{*}{\textbf{Sampling}} & \multirow{2}{*}{\textbf{Subset Size}}
    & \multicolumn{3}{c}{\textbf{NQ}}
    & \multicolumn{3}{c}{\textbf{MS-MARCO}}
    & \multicolumn{3}{c}{\textbf{HotpotQA}} \\
    \cmidrule(lr){3-5} \cmidrule(lr){6-8} \cmidrule(lr){9-11}
    \textbf{Number}& & ACC & ASR & PAD & ACC & ASR & PAD & ACC & ASR & PAD \\
    \midrule
    \multirow{3}{*}{200}
    & 3
      & 0.57 & 0.01 & 1.08
      & 0.74 & 0.00 & 1.21
      & 0.48 & 0.02 & 1.15 \\ 
    & 4
      & 0.64 & 0.00 & 0.94
      & 0.70 & 0.00 & 1.07
      & 0.56 & 0.02 & 1.00 \\ 
    & 5
      & 0.69 & 0.00 & 0.83
      & 0.76 & 0.00 & 0.96
      & 0.53 & 0.02 & 0.90 \\ 
    \midrule
    \multirow{3}{*}{400}
    & 3
      & 0.56 & 0.00 & 1.07
      & 0.69 & 0.00 & 1.21
      & 0.51 & 0.01 & 1.14 \\ 
    & 4
      & 0.67 & 0.00 & 0.93
      & 0.70 & 0.01 & 1.06
      & 0.54 & 0.03 & 0.99 \\ 
    & 5
      & 0.68 & 0.01 & 0.83
      & 0.79 & 0.00 & 0.96
      & 0.58 & 0.01 & 0.90 \\ 
    \midrule
    \multirow{3}{*}{600}
    & 3
      & 0.57 & 0.01 & 1.07
      & 0.67 & 0.00 & 1.21
      & 0.48 & 0.01 & 1.14 \\ 
    & 4
      & 0.62 & 0.00 & 0.93
      & 0.74 & 0.00 & 1.05
      & 0.52 & 0.02 & 0.99 \\ 
    & 5
      & 0.70 & 0.00 & 0.83
      & 0.73 & 0.01 & 0.95
      & 0.54 & 0.02 & 0.90 \\ 
    \bottomrule
    \end{tabular}}
    \label{table: sampling_ablation_1}
\end{table}

\subsection{\textcolor{blue}{Evaluation of the Impact of Monte Carlo Sampling}}
\label{sec:Monte Carlo}

\textcolor{blue}{
As the number of retrieved documents (Top-$K$) and the subset size ($n$) increase, the number of candidate combinations grows rapidly, incurring significant computational overhead. To reduce costs and enhance the practicality of our method, we introduce Monte Carlo sampling when the combination size reaches a certain level. Based on the analysis in Section~\ref{sec:appendix_sampling_proof}, we set the sampling count to $m=200$, which also serves as the threshold for triggering the sampling process. Specifically, when the number of candidate combinations does not exceed $m$, we employ the original full enumeration strategy; when the number exceeds $m$, we utilize Monte Carlo sampling to approximate the combinatorial space. This strategy effectively controls computational overhead while minimizing the impact on performance.
}

\renewcommand{\arraystretch}{1.2} % 默认是 1.0，调大数值可增大行距
\begin{table}[t]
    \centering
    \footnotesize
    % \color{blue}
    % \captionsetup{labelfont={color=blue}, textfont={color=blue}}
    \caption{Computational overhead of PRA-RAG using Monte Carlo sampling under varying Top-$K$ (subset size $n=3$, poisoning rate = 20\%).}
    \resizebox{0.99\linewidth}{!}{%
    \begin{tabular}{lcccc}
        \toprule
        \textbf{Models} & \textbf{Top-$K$} & \textbf{NQ} & \textbf{MS-MARCO} & \textbf{HotpotQA} \\
        \midrule
        \multirow{3}{*}{Mistral-7B} 
        & 8  & 5.98s & 6.20s & 6.17s \\
        & 12 & 11.78s & 11.71s & 11.97s \\
        & 16 & 11.09s & 11.69s & 12.50s \\
        \midrule
        \multirow{3}{*}{Vicuna-7B} 
        & 8  & 3.37s & 2.91s & 3.02s \\
        & 12 & 6.67s & 7.55s & 6.42s \\
        & 16 & 9.23s & 7.84s & 6.46s \\
        \bottomrule
    \end{tabular}}
    \label{table:cost_topk}
\end{table}

\renewcommand{\arraystretch}{1.2} % 默认是 1.0，调大数值可增大行距
\begin{table}[t]
    \centering
    \footnotesize
    % \color{blue}
    % \captionsetup{labelfont={color=blue}, textfont={color=blue}}
    \caption{Computational overhead of PRA-RAG using Monte Carlo sampling with Top-$K = 12$ under a poisoning rate of 10\%.}
    \resizebox{0.99\linewidth}{!}{%
    \begin{tabular}{lcccc}
        \toprule
        \textbf{Model} & \textbf{Subset Size} & \textbf{NQ} & \textbf{MS-MARCO} & \textbf{HotpotQA} \\
        \midrule
        \multirow{3}{*}{Mistral-7B} 
        & 3 & 12.26s & 10.51s & 11.93s \\
        & 4 & 11.96s & 15.07s & 12.16s \\
        & 5 & 11.95s & 15.09s & 12.17s \\
        \midrule
        \multirow{3}{*}{Vicuna-7B} 
        & 3 & 6.87s & 9.03s & 6.79s \\
        & 4 & 7.03s & 7.70s & 6.70s \\
        & 5 & 7.17s & 8.09s & 7.28s \\
        \bottomrule
    \end{tabular}}
    \label{table:cost_subset}
\end{table}

\textcolor{blue}{
% Tables~\ref{table: topk} and~\ref{table: n} present the performance of our method using Monte Carlo sampling approximation across various settings of retrieved document counts (Top-$K$) and subset sizes ($n$), specifically when the number of candidate combinations exceeds the threshold ($m=200$). By comparing these results with those obtained via full enumeration in Tables~\ref{table: topk-1} and~\ref{table: n-1}, it is evident that the introduction of Monte Carlo sampling causes no significant degradation in performance. Furthermore, the overall trends remain consistent, indicating that the impact of poisoned texts is further attenuated as Top-$K$ and $n$ increase, while PAD exhibits a continuous downward trend.
Tables~\ref{table: topk} and~\ref{table: n} present the performance of our method using Monte Carlo sampling approximation across various settings of retrieved document counts (Top-$K$) and subset sizes ($n$), specifically when the number of candidate combinations exceeds the threshold ($m=200$). By comparing these results with those obtained via full enumeration in Figure~\ref{fig:parameter}, it is evident that the introduction of Monte Carlo sampling causes no significant degradation in performance. Furthermore, the overall trends remain consistent, indicating that the impact of poisoned texts is further attenuated as Top-$K$ and $n$ increase, while PAD exhibits a continuous downward trend. To further strengthen this finding, we additionally conduct experiments with $K=20$ and subset sizes $n \in {3,4,5}$ under sampling counts $m \in {200,400,600}$. The results, reported in Table~\ref{table: sampling_ablation_1}, further confirm that the approximation quality remains stable as $K$ and $n$ scale.
}

\textcolor{blue}{
Tables~\ref{table:cost_topk} and~\ref{table:cost_subset} present the computational overhead of the algorithm with Monte Carlo sampling under different Top-$K$ and $n$ settings. In Table~\ref{table:cost_topk}, with Top-$K=8$ and $n=3$, the number of candidate combinations is 56, which remains below the threshold ($m=200$), implying that sampling is not triggered and the computational cost remains consistent with the results in Table~\ref{table:cost_topk-1}. In other scenarios, compared to the full enumeration results in Tables~\ref{table:cost_topk-1} and~\ref{table:cost_subset-1}, the sampling strategy significantly reduces computational overhead. For instance, on the Vicuna-7B model using the NQ dataset, the cost drops from 64.93s to 9.23s when Top-$K=16$ and $n=3$, and from 116.09s to 7.17s when Top-$K=12$ and $n=5$.
}
\end{document}